\documentclass[conference]{IEEEtran}
\usepackage{xfp}
\usepackage{xfrac}
\usepackage{siunitx}
\usepackage{etoolbox}
    \newrobustcmd\ubold{\DeclareFontSeriesDefault[rm]{bf}{b}\bfseries}

\usepackage{multirow}
\usepackage{amsmath}
\usepackage{amsfonts}
\usepackage{amsthm}
\usepackage{algorithm}
\usepackage{algorithmic}
\usepackage{graphicx}
\usepackage{balance}
\usepackage{xcolor}
\usepackage[numbers,sort&compress]{natbib}

\usepackage{hyperref}

\newtheorem{definition}{Definition}
\newtheorem{theorem}{Theorem}
\newtheorem{proposition}{Proposition}

\makeatletter
\newcommand{\thickhline}{%
    \noalign {\ifnum 0=`}\fi \hrule height 1pt
    \futurelet \reserved@a \@xhline
}
\makeatother

\newcommand{\Hide}[1]{}
\let\ignore\Hide

\DeclareMathOperator{\polylog}{polylog}

\begin{document}

\title{Faster Secure Comparisons with Offline Phase \\ for Efficient Private Set Intersection}


%
%

\author{\IEEEauthorblockN{Florian Kerschbaum}
\IEEEauthorblockA{University of Waterloo\\
Canada \\
florian.kerschbaum@uwaterloo.ca}
\and
\IEEEauthorblockN{Erik-Oliver Blass}
\IEEEauthorblockA{Airbus\\
Germany \\
erik-oliver.blass@airbus.com}
\and
\IEEEauthorblockN{Rasoul Akhavan Mahdavi}
\IEEEauthorblockA{University of Waterloo\\
Canada \\
rasoul.akhavan.mahdavi@uwaterloo.ca}}

\ignore{
\IEEEoverridecommandlockouts
\makeatletter\def\@IEEEpubidpullup{6.5\baselineskip}\makeatother
\IEEEpubid{\parbox{\columnwidth}{
    Network and Distributed System Security (NDSS) Symposium 2023\\
    28 February - 4 March 2023, San Diego, CA, USA\\
    ISBN 1-891562-83-5\\
    https://dx.doi.org/10.14722/ndss.2023.23xxx\\
    www.ndss-symposium.org
}
\hspace{\columnsep}\makebox[\columnwidth]{}}
}

\maketitle

\begin{abstract}
In a Private section intersection (PSI) protocol, Alice and Bob compute the intersection of their respective sets without disclosing any element not in the intersection.
  PSI protocols have been extensively studied in the literature and are deployed in industry.
With state-of-the-art protocols achieving optimal asymptotic complexity, performance improvements are rare and can only improve complexity constants.
In this paper, we present a new private, extremely efficient comparison protocol that leads to a PSI protocol with low constants.
A useful property of our comparison protocol is that it can be divided into an online and an offline phase.
All expensive cryptographic operations are performed during the offline phase, and the online phase performs only four fast field operations per comparison.
This leads to an incredibly fast online phase, and our evaluation shows that it outperforms related work, including KKRT (CCS'16), VOLE-PSI (EuroCrypt'21), and  OKVS (Crypto'21).
We also evaluate standard approaches to implement the offline phase using different trust assumptions: cryptographic, hardware, and a third party (``dealer~model'').
\end{abstract}






\section{Introduction}
\label{sec:intro}

Companies collect increasingly larger amounts of data about their
customers' operation. Each company collects different data depending
on their business, and the combination of these different data sets
offers greater benefit than each set by itself.
A standard example is Google collecting which user clicks on which online ad while
Mastercard collects financial transactions performed by its clients using their cards. 
To allow Google to compute the number of successful transactions after a user clicked on an online ad, Google and Mastercard link their data based on a common user identifier, e.g., the user's phone number.

Abstractly, this is an instance of Private Set Intersection (PSI).
In PSI, two parties, each have a set of (unique) elements and want to compute their intersection without revealing any element not in the intersection.
PSI is indeed deployed by Google and Mastercard to analyze ad conversions \cite{psigoogle20,psigoogle15}, but it has many more applications.
Consequently, PSI has recently been extensively studied in the literature~\cite{authpsi,certifiedsets,psigoogle15,psigoogle20,opprf,cm,kkrt,spot,psiot,phasing,vole,freedman,meadows,facebook,secsharedPSI,yaoPSI,imbalancePSI,imbalance,fhePSI1,fhePSI2,origoopprf,dualcuckoo,otoprf,malicious,paxos,hashotex,rbc1,rbc2,outsourcedPSI,delegated,delegatedcard,arbiterPSI,mulPSI,threshold,psi04,googlemal,fasterfhepsi,pets22,crypto21}.
The currently most efficient state-of-the-art PSI protocols are based on oblivious pseudo-random functions~\cite{kkrt,vole,cm}.
They require a constant number of public-key cryptography operations, linearly many symmetric key cryptography operations, and one round of interaction.
This is asymptotically optimal, and any performance improvement can only stem from reduced constants which are, however, already very low.
We note that PSI requires public-key operations, since two-party computation can be reduced to PSI (see our reduction in Section~\ref{sec:reduction}), and two-party computation requires public-key operations.
Hence, any (new) approach must deal with these unavoidable and expensive operations.

In this paper, we present a new (equality) comparison protocol that is simple, elegant, and very efficient.
Our comparison protocol improves over the state-of-the-art in two aspects:
first, it reduces (equality) comparison to {\em oblivious linear evaluation} (OLE), and, second, it enables the use of offline precomputed OLE tuples instead of computing the OLE online.
This results in an alternative construction of PSI that off-loads all expensive cryptographic operations, public and symmetric key, to an initial offline phase.
The offline phase precomputes correlated randomness and can be run in advance, independently of the inputs to the protocol.
The online phase uses this randomness and the inputs to securely compute the output.
Our online phase is highly efficient, comprising only four fast operations in a small field, i.e., one multiplication and three additions per comparison. Moreover, it takes only one round.
The offline phase can be implemented using standard approaches based on only cryptographic assumptions, e.g., using lattice-based homomorphic encryption or Oblivious Transfer (OT), but also based on more efficient hardware or other trust assumptions, such as a trusted third party (``dealer model''~\cite{hastings}).


\subsection{Why consider an offline phase?}
\label{sec:whyoffline}

Our offline phase enables a very fast online phase.
In Figure~\ref{fig:over} we display a comparison of computation time and communication cost to related work \cite{cm,kkrt,spot,vole,crypto21} for a data set size of 1 million elements {\em in the online phase online}.
Note that VOLE-PSI~\cite{vole} requires less communication for these small sets.
On a larger data set with 16 million elements our online phase is between $2.4$ and $3.5$ times faster than Kolesnikov et al.'s work~\cite{kkrt} and $1.2$ and $14.6$ times faster than \citeauthor{vole}'s work~\cite{vole} while requiring less communication than either one of them.
In general, our advantage increases as data set sizes grow, but \citeauthor{vole} neither make their code available nor do they report evaluation results for larger data set sizes which makes comparison difficult.

\begin{figure}[!t]\centering
  \begin{center}
    \input{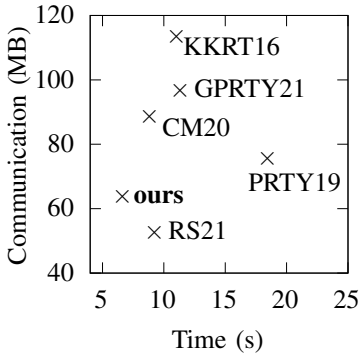}
  \end{center}
  \caption{\label{fig:over}Time (online) and communication of our protocol vs  state of the art (KKRT16~\cite{kkrt}, PRTY19~\cite{spot}, CM20~\cite{cm}, RS21~\cite{vole}, and GPRTY21~\cite{crypto21}), sets of size $n=2^{20}$, $100$~MBit/s bandwidth.}
\end{figure}

Employing an offline phase is common in two-party computation, even if the total run-time (slightly) increases, and also in the context of PSI, it is not new.
We emphasize that the commonly most efficient related work, VOLE-PSI by Rindal and Schoppmann~\cite{vole}, already requires an offline phase to derive vector OLE tuples and has been developed by industry, thus demonstrating the acceptability of such an offline phase.
The type of  offline phase we employ has further advantages, since it can be efficiently implemented using alternatives to cryptographic assumptions and similar code for the offline phase is re-usable among secure computation protocols based on Beaver multiplication triples~\cite{beaver}. This includes the prominent SPDZ family of protocols~\cite{spdz,mascot,overdrive} for generic multi-party computation.

In summary, we believe that there are good reasons to consider our construction of a PSI protocol using an offline phase.
There already exist many protocols for the offline phase, e.g., for secure two-party computation, and our work reuses those protocols in its construction or implements adjustments on top of them.
This has the additional advantage that further research in improving these commonly used protocols will readily extend to our construction.
Given additional security assumptions, such as the availability of trusted execution environments during the offline phase (only), our (entire) protocol reduces computation and communication costs compared to prior work.

\subsection{Contributions}
This paper contributes a new private set intersection protocol based on a new comparison protocol that can be divided into an {\em offline} and {\em online} phase.
In particular, we provide
\begin{itemize}
    \item a new, highly efficient {\em online} phase of our PSI/comparison protocol that uses no cryptographic operations, but uses precomputed, random {\em OLE tuples}.
    \item several adaptations of existing protocols for the {\em offline} phase, based on lattice-based homomorphic encryption, oblivious transfer, trusted execution environments, and a trusted third party.
    \item formal security proofs of our protocols.
    \item an extensive experimental evaluation of our protocol in comparison to related work \cite{cm,kkrt,spot,vole,crypto21}.
\end{itemize}

Our PSI protocol's online phase is $1.2$ (10 MBit/s network) to $3.5$ (5 GBit/s network) times faster than the currently most efficient related work \cite{kkrt,spot,cm,vole,crypto21}.
The maximum performance increase of our protocol in total time shrinks to $1.5$ when used with an offline phase based on additional hardware trust assumptions not made in related work.

The next section contains preliminary definitions and building blocks used in our protocols.
In Section~\ref{sec:psi} we describe the online phase of our comparison and PSI protocol.
In Section~\ref{sec:precomputation} we describe different protocols for the offline phase.
We present our experimental results in Section~\ref{sec:evaluation}.
We compare our work to related work in Section~\ref{sec:related} and present our conclusions in Section~\ref{sec:conclusions}.

\section{Preliminaries}
\label{sec:prelim}


\subsection{Semi-Honest Security Model}
\label{sec:semihonest}

We consider two-party PSI computation in the semi-honest, or passive, security model.
In this model, parties Alice and Bob are assumed to follow the protocol as prescribed but keep a record of each message received, their random coins, and their input.
From this information, called the view $\mathsf{View}_X$ of a party $X \in \{ A, B \}$, they try to compute additional information beyond the output.
Note that the output of a party $X$ can be computed from its view $\mathsf{View}_X$.
Informally, a protocol is secure in the semi-honest security model, if a party cannot compute additional information.
This can be proven by showing the existence of a simulator $\mathsf{Sim}_X$ that given the party's input $x$ or $y$ and the respective output $\chi$ or $\psi$ produces a simulation of the view that is indistinguishable from the party's view $\mathsf{View}_X$ during the protocol.

\begin{definition}
Protocol $\pi$ is secure in the {\em semi-honest model}, if there exist simulators $\mathsf{Sim}_A(x, \chi)$ and $\mathsf{Sim}_B(y, \psi)$, such that
\begin{eqnarray*}
\mathsf{View}^{\pi}_A(x, y) & = & \mathsf{Sim}_A(x, \chi) \\
\mathsf{View}^{\pi}_B(x, y) & = & \mathsf{Sim}_B(y, \psi)
\end{eqnarray*}
\end{definition}

The semi-honest security model assumes that an adversary behaves passively, i.e., does not modify its messages.
The malicious security model removes this assumption and considers active adversaries behaving arbitrarily.
Any protocol secure in the semi-honest model can be compiled into one secure in the malicious model using the GMW compiler \cite{gmw}, but this compilation may lead to a very inefficient protocol.
Hence, many previous works on PSI consider the semi-honest model.
The malicious model for PSI has two limitations:
First, since the intersection is revealed, a malicious party may simply substitute its input and learn (parts of) the other party's set which cannot be prevented in the malicious model.
Second, in one-round protocols, such as ours, arbitrary behaviour is only possible in the input-carrying first message.
Hence, additional information can only be computed from the message (output) received in response to this first message.
Techniques such as certification of inputs~\cite{certifiedsets,authpsi} can mitigate input substitution attacks.
Our protocol can be augmented with input certification.
However, we leave a detailed description to future work.

\subsection{Oblivious Transfer}
\label{sec:ot}

Oblivious transfer (OT) \cite{ot} is a protocol between two parties Alice and Bob.
Its simplest variant is a 1-out-of-2 OT \cite{1outof2ot}.
Alice has input $b \in \{0, 1\}$.
Bob has input $y_0, y_1$.
An OT protocol is {\em correct} if Alice obtains $y_b$.
An OT protocol is {\em secure} if Alice learns nothing about $y_{1-b}$, and Bob learns nothing about $b$.
OT is a powerful primitive and can be used to implement secure two-party computations~\cite{gmw,kilian,foundingoteff}.
Hence, any OT protocol requires at least one public-key operation, since secure two-party computations require at least one public-key operation.

\subsection{Oblivious Linear Evaluation}
Oblivious Linear Evaluation (OLE) is a building block in many secure
computation protocols~\cite{vole} and a method to generate correlated
randomness between two parties.  From the slightly different OLE
definitions in the literature~\cite{applebaum,ole,vector-ole,vole,compress}, we use the
following.

Oblivious Linear Evaluation is a secure two-party computation protocol between parties Alice and Bob, such that Bob samples $a,b\in \mathbb{F}$ for some field $\mathbb{F}$, and Alice samples $u\in\mathbb{F}$. After the protocol, Alice obtains $f(u)$ such that
\begin{align}\label{ole-def}
    f(u) = au + b.
\end{align}

The symmetric variant of OLE is referred to as product sharing~\cite{ole}, where Alice and Bob each know $u$ and $v$, respectively, and they obtain additive shares of the product. In other words, after the protocol, Alice obtains $a$ and Bob $b$, such that
\begin{align}\label{prod-share-def}
    uv = a + b.
\end{align}
As opposed to OLE, $a$ and $b$ are generated randomly. Product sharing can be used to construct OLE protocols~\cite{ole}.

In vector OLE (vOLE), which is a generalization of OLE, Bob knows two vectors  $A,B\in \mathbb{F}^{n}$, and Alice knows $u\in\mathbb{F}$ and obtains $V\in\mathbb{F}^{n}$ where
\begin{align}
    V = Au + B.
\end{align}

Similarly in batch OLE (bOLE), Alice also samples a vector $U\in\mathbb{F}^{n}$ and obtains $V\in\mathbb{F}^{n}$ where 
\begin{align}
    V = A \cdot U + B,
\end{align}
and the multiplication is element-wise.
Note that, similarly to the notion of random OT, there exist also random alternatives of the above OLE protocols.  So, parties receive random tuples $(a,u)$ and $(b,v)$ with $uv=a+b$.

In this work, we require a slightly different variant than (random) vOLE and bOLE for our PSI protocol. 
Parties will compute random batches of tuples, but the type of batches and choice of correlated
randomness will be different. We will present details in Section~\ref{sec:modbeaver} later.

OLE was first formalized by Applebaum et al.~\cite{applebaum} where they proposed that OLE is the arithmetic analog of oblivious transfer.
One approach for OLE over an arithmetic field is to repeatedly invoke 1-out-of-2 bit-OT to retrieve the bits of the results one by one. 
This approach has a computational complexity proportional to the size of the binary circuit that evaluates the linear function~\cite{product-sharing}.
Another approach is to use hardness assumptions in coding theory such as the pseudo-randomness of noisy random codewords in a linear code or the LPN problem~\cite{Naor2006ObliviousPE}.

Finally, protocols realizing OLE based on lattices and lattice-based encryption schemes are also common~\cite{ole,fast-ole}.
All of these approaches (except LPN) are used in this work to implement our offline phase.

\subsection{Reduction of Two-Party Computation to PSI}
\label{sec:reduction}

\begin{table}[!h]
\caption{Intersections $C\cap{}D$ for the different inputs by Alice and Bob.}
\label{tbl:reduction}
\centering
\begin{tabular}{c||c|c|c|c}
\multicolumn{1}{c||}{\multirow{4}{*}{Alice}} & \multicolumn{4}{c}{Bob} \\
\multicolumn{1}{c||}{} & \multicolumn{2}{c|}{$y_0 = 0$} & \multicolumn{2}{c}{$y_1 = 1$} \\
\multicolumn{1}{c||}{} & $y_1 = 0$ & $y_1 = 1$ & $y_1 = 0$ & $y_1 = 1$ \\
\multicolumn{1}{c||}{} & $D = \{2, 3\}$ & $D = \{1, 2\}$ & $D = \{0, 3\}$ & $D = \{0, 1\}$ \\
\hline
\hline
\begin{tabular}{c}
     $b=0$\\
     $C = \{0\}$
\end{tabular}
 & $\emptyset$ & $\emptyset$ & $\{0\}$ & $\{0\}$ \\
\hline
\begin{tabular}{c}
     $b=1$\\
     $C = \{1\}$
\end{tabular}
& $\emptyset$ & $\{1\}$ & $\emptyset$ & $\{1\}$ \\
\end{tabular}
\end{table}

To underpin the power of PSI, we now present a reduction of two-party computation to PSI.
This reduction may be folklore, but we have not found it spelled out in the literature.
We prove that the existence of one-sided PSI, where only Alice learns the intersection, implies the existence of OT by reducing OT to PSI.
It also proves that public-key operations are necessary for one-sided PSI since they are necessary for oblivious transfer and two-party computation.
The GMW protocol \cite{gmw} over the binary field $\mathbb{Z}_2$ reduces (semi-honest) two-party and multi-party computation to OT.
Reductions of maliciously secure two-party computation also exist \cite{kilian,foundingoteff}.

We compute a 1-out-of-2 OT for one-bit messages.
Alice has $b \in \{ 0, 1 \}$.
Bob has two messages $y_0, y_1 \in \{0, 1\}$.
Alice should obtain $y_{b}$, but not learn anything about $y_{1 - b}$.

If $b = 0$, Alice chooses $C = \{ 0 \}$ as her input set.
If $b = 1$, Alice chooses $C = \{ 1 \}$.
Bob starts by creating an empty set $D$.
If $y_0 = 0$, Bob adds $\{ 2 \}$ to $D$.
If $y_0 = 1$, Bob adds $\{ 0 \}$ to $D$.
If $y_1 = 0$, Bob adds $\{ 3 \}$ to $D$.
If $y_1 = 1$, Bob adds $\{ 1 \}$ to $D$. 

Alice and Bob perform a one-sided PSI protocol for sets $C$ and $D$, such that Alice learns intersection $C\cap{}D$.
Table~\ref{tbl:reduction} shows the resulting intersection.
Observe that the size of input sets is constant ($|C| = 1$, $|D| = 2$).
If the intersection $C\cap{}D$ is empty, then Alice outputs $0$.
If the intersection is either $\{0\}$ or $\{1\}$, Alice outputs $1$.

Note that privately evaluating functions over the intersection is sufficient for this reduction, but not necessary, e.g., set intersection cardinality or set disjointness.
Alice only obtains either of two possible outputs, either the empty set or her input set of size $1$.

Furthermore, OT can be reduced to labelled PSI -- an extended form of PSI.
In labelled PSI Alice obtains a message (label) for her elements in the set.
Bob simply sets the labels $y_i$ for elements $i = \{ 0, 1 \}$.
However, it does not seem obvious how to reduce labelled PSI to PSI which our reduction also implies.

\section{Online Phase}
\label{sec:psi}
For ease of exposition, we begin by  describing the online phase of our comparison and PSI protocols.

{\noindent{\bf Roadmap for this section.}}
For now, assume that the offline phase has output so called \emph{OLE tuples} which we define in Section~\ref{sec:modbeaver} and show how to compute during the offline phase in Section~\ref{sec:precomputation}.
Given an OLE tuple, we will present how to compare two elements $x$ held by Alice and $y$ held by Bob (Section~\ref{sec:comparison}).
Finally, we will use hashing techniques to construct a full PSI protocol for sets $\mathbb{X}$ held by Alice and $\mathbb{Y}$ held by Bob, each of size $n$, in Section~\ref{sec:protocol}.
We assume that the sets have the same size for simplicity of exposition and the analysis of the algorithm.
If they do not, the protocol's cost is dominated by the larger one as it is for any of the compared related work.
The goal of the full PSI protocol is to reduce the number and cost of necessary computations.
The online phase requires $O(n \log n / \log \log n)$ comparisons, but our constants are so low that, in concrete numbers, we improve over related work, most of which also require at least $O(n \log n / \log \log n)$ computational steps.

\subsection{OLE Tuples}
\label{sec:modbeaver}

We operate over a prime field $\mathbb{F}_Q$.
Let $\mathbb{F}^*_Q = \mathbb{F}_Q \setminus \{0\}$, and let $r_B$ be a random number in $\mathbb{F}^*_Q$.
Similarly, let $r_A, s_A, s_B$ be three random numbers in $\mathbb{F}_Q$.
The marginal distribution of each random number is uniform, but their joint distribution is correlated and satisfies
\begin{equation}
\label{eq:mbt}
r_A r_B = s_A + s_B
\end{equation}
If $s_A + s_B = 0$, then $r_A = 0$, but it always holds that $r_B \neq 0$.
In our implementation of the offline phase $s_A$, $s_B$, and $r_B$ are drawn independently, and $r_A$ is computed as correlated randomness from their choices to satisfy Equation~\ref{eq:mbt}.

Such a random tuple $(r_A,r_B,s_a,s_B)$ is called an \emph{OLE tuple}.
For our PSI online phase, we assume that the offline phase
has generated a sufficient number of OLE tuples and distributed the $(r_A,s_A)$ to Alice and the $(r_B,s_B)$ to Bob.

Observe the similarity of Equation~\ref{eq:mbt} to Equation~\ref{ole-def}.
If we set $r_A = u$, $r_B = a$, $s_B = -b$, and $s_A = f(u)$, we obtain a re-ordering of Equation~\ref{eq:mbt}.
As in OLE, $r_A$ and $s_A$ ($u$ and $f(u)$) are known only to Alice, and $r_B$ and $s_B$ ($a$ and $b$) are known only to Bob. 

Note that, in Section~\ref{sec:commopti}, we will also introduce a new {\em set version} of OLE tuples that further optimizes our online communication cost.
However, this optimization also comes at a cost:
We then need to adapt existing protocols in the offline phase for our set OLE tuples.

\subsection{Comparison Protocol}
\label{sec:comparison}
We describe the comparison of a single pair of input elements and then extend it to a full-fledged PSI protocol in the next section.
For some input length $\sigma\in\mathbb{N}$, Alice has element $x\in\{0,1\}^{\sigma}$, and Bob has element $y\in\{0,1\}^{\sigma}$.
As part of the PSI protocol, they want to determine whether $x=y$.
Let $H: \{0, 1 \}^{\sigma} \mapsto \mathbb{F}_Q$ be a cryptographic hash function.
In case the element's bit length is shorter than the hash's bit length ($\{0, 1 \}^{\sigma} \subset \mathbb{F}_Q$), we do not need a hash function.

The comparison consumes one OLE tuple $(r_A,r_B, s_A,s_B)$, where $(r_A,s_A)$ is know to Alice and $(r_B,s_B)$ is known to Bob, and works as follows. 
\begin{enumerate}
\item Alice starts by computing $c = s_A - H(x)$ and sending $c$ to Bob.
\item Bob computes 
\begin{equation}
d = (c + H(y) + s_B) / r_B
\label{eq:compare}
\end{equation}
and sends $d$ back to Alice.
\item Alice verifies whether $d \stackrel{?}{=} r_A$, and if they are equal, she outputs ``match''.
\end{enumerate}
  Note that we do not need to compute the multiplicative inverse of $r_B$.
We can immediately choose it when generating the OLE tuples (see Section~\ref{sec:precomputation}).

\begin{theorem}
\label{thm:correct}
Our comparison protocol is {\em correct}, i.e., if and only if $x = y$, then Alice outputs ``match'',  assuming no hash collisions.
\end{theorem}

\begin{proof}
We can substitute $c$ in Equation~\ref{eq:compare} and obtain:
\begin{equation}
\label{eq:comp}
d = (s_A - H(x) + H(y) + s_B) / r_B
\end{equation}

We show perfect correctness, i.e., if $x = y$, then $d = r_A$ and if $x \neq y$, then $d \neq r_A$.
We can reformulate Equation~\ref{eq:comp} as
\begin{equation*}
\label{eq:chi}
d r_B = s_A - H(x) + H(y) + s_B
\end{equation*}
and subtract Equation~\ref{eq:mbt}
\begin{equation*}
(d - r_A) r_B = H(y) - H(x)
\end{equation*}
Since $r_B \in \mathbb{F}^*_{Q} \neq 0$, it follows
\begin{equation*}
d - r_A = 0 \Leftrightarrow H(y) - H(x) = 0
\end{equation*}
and then
\begin{equation*}
d = r_A \Leftrightarrow H(x) = H(y)
\end{equation*}\end{proof}

\begin{theorem}
\label{thm:secure}
Our comparison protocol is {\em secure}, i.e., there exists a simulator of Bob's and Alice's view.
\end{theorem}

\begin{proof}
{\em Bob's simulator}:
Since $r_A$ and $s_A$ are unknown to Bob, and he only learns Equation~\ref{eq:mbt}, $s_A$ is uniformly distributed in $\mathbb{F}_Q$ for Bob.
Consequently, $c$ is uniformly distributed in $\mathbb{F}_Q$ for Bob and the simulator can output a uniformly chosen random number.

{\em Alice's simulator}:
If $x = y$, then the simulator outputs $d = r_A$.
We have shown perfect correctness in Theorem~\ref{thm:correct} and hence only need to deal with the case $d \neq r_A$ ($x \neq y$).

We show that since the set $\{r_B, s_B\}$ is unknown to Alice\footnote{Note that if $r_A=0$, Alice knows $s_B = -s_A$, but $r_B$ is still unknown to Alice which suffices for the proof.}, $d$ is uniformly distributed in $\mathbb{F}_Q\setminus{}\{r_A\}$ for Alice, if $x \neq y$.


Let
\begin{equation}
\label{eq:mu}
\mu = H(x) - H(y)
\end{equation}
Since $x \neq y$ (and we assume no collisions), it holds that $\mu \in \mathbb{F}^*_Q$.
Substituting Equation~\ref{eq:mu} into Equation \ref{eq:comp} we get
\begin{equation}
\label{eq:sub1}
d = (s_A + s_B + \mu) / r_B
\end{equation}
Solving Equation~\ref{eq:mbt} for $r_A$, subtracting it from Equation~\ref{eq:sub1}, and solving for $d$ we get
\begin{equation}
\label{eq:sub2}
d = \mu / r_B + r_A
\end{equation}
The simulator needs to output $d$, but is given $r_A$ as input.
$\mu$ is in $\mathbb{F}^*_Q$ and $r_B$ is uniformly distributed in $\mathbb{F}^*_Q$.
$\mathbb{F}^*_Q$ forms the multiplicative sub-group in $\mathbb{F}_Q$.
Hence, $\mu / r_B$ is uniformly distributed in $\mathbb{F}^*_Q$.
The simulator outputs the sum of a uniform random number $v$ in $\mathbb{F}^*_Q$ and $r_A$.
Since $\mu / r_B \neq 0$, it holds $d \neq r_A$.
In summary, the simulator chooses a uniform random number in $\mathbb{F}_Q\setminus{}\{r_A\}$.
\end{proof}

{\noindent{\bf OPPRFs.}} An oblivious {\em programmable} pseudo-random function (OPPRF) \cite{kkrt,opprf} is an oblivious pseudo-random function (OPRF) that can be programmed for a number of input values.
Interestingly, our proofs show that our comparison protocol can be interpreted as an OPPRF that can be programmed for a {\em single} input.
However, our comparison protocol is much faster than OPPRFs for multiple inputs \cite{kkrt,opprf} which explains its advantage over the state-of-the-art even when the comparison protocol is repeated multiple times.

\subsection{Full PSI Protocol}
\label{sec:protocol}

In the full PSI protocol, Alice and Bob want to compare elements in the sets $\mathbb{X} = \{ x_i \}$ held by Alice and $\mathbb{Y} = \{ y_j \}$ held by Bob.
In order not to compare each pair of elements $(x_i,y_j)$, we use a well-known technique of hashing elements to bins and only comparing elements within each bin.
Alice hashes her elements $x_i \in \mathbb{X}$ using cuckoo hashing with $k$ hash functions and hashes into $\alpha$ bins.
Bob uses regular hashing for his set $\mathbb{Y}$ with each of the $k$ hash functions into the same hash table~\cite{psiot,phasing,opprf,freedman}.

\subsubsection{Alice: Cuckoo Hashing}
Cuckoo hashing \cite{cuckoo} is a hashing technique that reduces the maximum number of elements per bin to $1$.
Cuckoo hashing uses $\alpha = O(n)$ bins and $k > 1$ hash functions $h_i$.
The idea of cuckoo hashing is to iterate over the $k$ hash functions for each element.
First, an element $x$ is inserted into bin $h_1(x)$.
However, if this bin is already occupied by element $y$, as $h_1(x)=h_i(y)$, then $x$ replaces $y$ in that bin, and $y$ will be inserted into bin $h_{i+1\bmod k}(y)$ and so on. 
There is a chance of an infinite loop of replacements, so the algorithm stops after a logarithmic number of replacements and places the current remaining element on a small stash data structure.
For our protocol, the size $s$ of the stash needs to be fixed as well.
Let  $\lambda$ be a statistical security parameter.
We choose $s$, such that the probability of failure, i.e., exceeding the stash size, is less than $2^{-\lambda}$.
In our experiments, we set $\lambda = 40$.
We use the parameter sets of Pinkas et al.~\cite{psiot} which we evaluate for our protocol in Section~\ref{sec:communication}.
In our optimal setting selected in Section~\ref{sec:communication} a stash is not even necessary.

Note that, in order to not leak information about her set, Alice pads empty bins with dummy elements, such that the total number of elements per bin is always $1$.

\subsubsection{Bob: Regular hashing}
Bob uses regular hashing to map his input into a regular hash table of length $\alpha$.
Since Alice may use any of the $k$ hash functions to bin element $x$, Bob has to use all $k$ hash functions for his input $y$, guaranteeing that $x$ and $y$ will be in a common bin, if $x = y$.
Inserting each element into (up to) $k$ bins in Bob's table, increases the total number $\beta = O(k \log n / \log \log n)$ \cite{gonnet} of elements per bin.
$\beta$ is chosen such that the probability that Bob will exceed $\beta$ in any hashtable bin is negligible in the statistical security parameter $\lambda$, i.e., smaller than $2^{\lambda}$.
Note that $\beta$ does not depend on Bob's input.
To not leak any information about Bob's set, he also has to pad all bins to $\beta$ elements.
Observe that $\beta$ is crucial for performance, as Alice and Bob have to perform $\beta$ comparisons for each of Alice's elements.
Hence, reducing $\beta$ significantly reduces the number of necessary comparisons (including those with dummy elements) and hence increases the overall performance of the protocol.
We use the parameter sets of Pinkas et al.~\cite{phasing} for $\beta$ which we evaluate for our protocol in Section~\ref{sec:communication}.

\subsubsection{Alice and Bob: Comparisons}

\begin{algorithm}[!t]
\caption{Our PSI protocol}
\label{alg:protocol}
\begin{algorithmic}
\STATE Common Input: A set $H$ of $k$ hash functions
\STATE Input {\em Alice}: $\mathbb{X}$, OLE tuples $\{ \{ r_{A,i,j} \}, s_{A,i} \}$ ($i\in[\alpha]$, $j\in[\beta]$)
\STATE Input {\em Bob}: $\mathbb{Y}$, OLE tuples $\{ \{ r_{B,j,i}, s_{B,j,i} \}\}$
\STATE Output {\em Alice}: $\mathbb{X} \cap \mathbb{Y}$
\STATE \hspace{5pt}
\STATE {\em Alice}:
\STATE $\alpha \gets (1 + \epsilon) \cdot |\mathbb{X}|$
\STATE $T_A \gets \{\mathsf{dummy}\}^\alpha$
\FORALL{$x_i \in \mathbb{X}$}
  \STATE Permutation-based cuckoo hash $x_i$ into $T_A$
\ENDFOR
\FORALL{$t_i \in T_A$}
  \STATE Send $i, c_i \gets s_{A,i} - t_i$ to Bob
\ENDFOR
\STATE \hspace{2pt}
\STATE {\em Bob}:
\STATE $T_B \gets \{ \{\mathsf{dummy}\}^\beta \}^\alpha$
\FORALL{$y_i \in \mathbb{Y}$}
  \FORALL{$h_j \in H$}
    \STATE Add permutation-based suffix of $y_i$ to $T_B$ using $h_j$
  \ENDFOR
\ENDFOR
\FORALL{$T_i \in T_B$}
  \FORALL{$t_{i,j} \in T_i$}
    \STATE Send $i, j, d_{i,j} \gets (c_i + t_{i,j} + s_{B,j,i}) / r_{B,j,i}$ to Alice
  \ENDFOR
\ENDFOR
\STATE \hspace{2pt}
\STATE{\em Alice}:
\FORALL{$d_{i,j}$}
  \IF{$d_{i,j} = r_{A,j,i}$}
    \STATE Output the $x$-value hashed into bin $i$
  \ENDIF
\ENDFOR

\end{algorithmic}
\end{algorithm}

Alice and Bob now need to run $\alpha \cdot \beta$ comparisons, i.e., $\beta$ comparisons for each of the $\alpha$ bins.
With a stash ($s > 0$), Alice and Bob  also compare each element in the stash with each element of $Bob$, i.e., $s \cdot n$ additional comparisons.
Each comparison uses the protocol from Section~\ref{sec:comparison}.
Note that all comparisons can be performed in parallel, and for each element in $\mathbb{X}$ or $\mathbb{Y}$, there can be at most one match.

While this already concludes the description of the main steps of the
full PSI protocol, we also institute the following crucial
optimizations.

\subsubsection{Permutation-Based Hashing} The comparison protocol operates over field $\mathbb{F}_{Q}$ and the up to $2n$ elements (from the union of $\mathbb{X}$ and $\mathbb{Y}$) need to be mapped to $\mathbb{F}_{Q}$, ideally without collision.
Hence, the size of $\mathbb{F}_{Q}$ is determined by the domain of $\mathbb{X}$ and $\mathbb{Y}$.
However, the size of $\mathbb{F}_{Q}$ is also linear in our communication complexity and any reduction again significantly reduces the overall communication cost.
Thus, Alice uses  permutation-based hashing \cite{permutationhash}, similar to the Phasing protocol \cite{phasing}, to further reduce the size of $\mathbb{F}_Q$.

Permutation-based hashing works as follows.
Let $x$ be an $\sigma$-bit string which consists of a prefix $x_1$ of $\sigma_1$ bits and a suffix $x_2$ of $\sigma_2$ bits, such that the concatenation of $x_1$ and $x_2$ is $x$.
Let $h$ be a hash function that maps prefixes to bins.
Element $x$ is inserted into bin $h(x_2) \oplus x_1$.
This ensures that if two different elements $x \neq y$ have a common suffix, i.e., $x_2 = y_2$ and consequentially $h(x_2) = h(y_2)$, they will be mapped to different bins, since then it must be that $x_1 \neq x_2$.
Hence, it suffices to compare $x_2$ and $y_2$ for each bin to determine equality of $x$ and $y$.
However, since we use cuckoo hashing, we need to ensure that the same hash function $h_i$ is used.
Otherwise, two different elements $x \neq y$ with the same suffix $x_2 = y_2$ could be mapped to the same bin $h_i(x_2) \oplus x_1 = h_j(y_2) \oplus y_1$, but have different prefixes $x_1 \neq y_1$ since $h_i \neq h_j$.
We ensure this by using a field size of $\mathbb{F}_Q$ that is larger than $k 2^{\sigma_2}$ and encoding suffixes into different ranges depending on the hash function, such that suffixes can only match if they use the same hash function.
When using $k$ hash functions, our technique increases the message length by $\log_2(k)$ bits.  
Permutation based hashing saves $\log_2(\alpha)$ bits.  
Hence, the net effect is $\log_2(\alpha) - \log_2(k)$ which is quite large for the parameters in Table~\ref{tab:para}.

\subsubsection{Communication Optimization for OLE Tuples}
\label{sec:commopti}
Using the above hashing techniques, Alice compares each of her (dummy and input) elements with $\beta$ (dummy or input) elements from Bob.
This allows for another optimization of Alice's (and thus the total) communication complexity.

Instead of using a new $s_A$ for each comparison, Alice re-uses one
$s_A$ for all $\beta$ comparisons of one $x$, but uses different
$r_A$.  Bob never re-uses a pair $(s_B, r_B)$. So, we batch one $s_A$
with $\beta$-many $r_A$, $s_B$, and $r_B$.

An \emph{optimized} OLE tuple is tuple $(s_A, (s_{B,1},r_{A,1},r_{B,1}),\ldots, (s_{B,\beta}$, $r_{A,\beta},r_{B,\beta}))$, where 
\begin{equation}
\label{eq:smbt}
\forall i \in \{ 1, \ldots, \beta \}: r_{A,i} r_{B,i} = s_A + s_{B,i}.
\end{equation}

For the online phase of the PSI protocol, we assume that the offline phase has generated a sufficient number of optimized OLE tuples by choosing $s_A,s_{B,i}\in\mathbb{F}_Q,r_{B,i}\in\mathbb{F}^*_Q$ randomly and $r_{A,i}\in\mathbb{F}_Q$ as correlated randomness such that Equation~\ref{eq:smbt} holds.
For each tuple the $(s_A,r_{A,1},\ldots,r_{A,\beta})$ are distributed to Alice and the $((s_{B,1},r_{B,1}),\ldots,(s_{B,\beta},r_{B,\beta}))$ to Bob.

In the remainder of the paper, we will only consider these optimized OLE tuples and show how to construct large amounts of such tuples in Section~\ref{sec:precomputation}.

The computation of these sets is also more efficient than the computation of all distinct tuples.
Overall, this reduces the number of messages sent by Alice in the online phase from $\alpha \cdot \beta$ to $\alpha$.

We provide a proof of security for this optimization.
Let $x$ be an element held by Alice and Bob hold a set of elements, such that Alice compares her element against each element in Bob's set as described.
We call this a set comparison protocol.

\begin{theorem}
Our set comparison protocol is {\em secure}, i.e., there exists a simulator of Bob's and Alice's view.
\end{theorem}

\begin{proof}

{\em Bob's simulator}:
Same as in Theorem~\ref{thm:secure}.

{\em Alice's simulator}:
The case $x = y_j$ is handled as above.
We consider the simulation for each element $y_j \neq x$ in Bob's set.
Let
\begin{equation*}
\mu_{j} = H(x) - H(y_{j})
\end{equation*}

Equation~\ref{eq:sub2} is modified to
\begin{equation*}
d_j = \mu_{j} / r_{B, j} + r_{A, j}
\end{equation*}
Since all $r_{B, j}$ are independently uniformly distributed in $\mathbb{F}^*_Q$, the same derivation as in the proof of Theorem~\ref{thm:secure} holds.
The simulator chooses $j$ uniform random numbers in $\mathbb{F}_Q\setminus{}\{r_A\}$ if $x \neq y_j$.
\end{proof}

{\noindent{\bf Formal Presentation.}} We formalize our full PSI protocol, including the above optimizations and assuming that there is no stash, in Algorithm~\ref{alg:protocol}.
It is a sequential composition of set comparison protocols, one for each bin in the hash table.

\section{Offline Phase (OLE Tuple Precomputation)}
\label{sec:precomputation}
We now turn to the various options to realize the offline phase for securely computing optimized (set) OLE tuples from Section~\ref{sec:commopti}.
These are (sometimes small) adaptations of existing state-of-the-art techniques that complete our comparison protocol into a full-fledged PSI protocol.

We consider purely cryptographic, hardware-supported, and trusted third-based approaches for the implementations of the offline phase.
For fair comparison, we note that the related work with which we compare our performance \cite{cm,kkrt,spot,vole,crypto21} only uses purely cryptographic assumptions.
However, other related works have considered our assumptions as well to speed up their (related) protocols.
We start with an overview.

\subsubsection{TTP}
The first option we consider is that of a trusted third party (TTP) providing OLE tuples
similar to Beaver triples as a service.  The implementation of Beaver
multiplication triples as a service by a TTP (``dealer
model'') to speed up their computation has been proposed multiple
times in the literature, see, e.g.,
\cite{tripleservice1,tripleservice2,hastings} for an overview.  Also,
some state-of-the-art PSI protocols assume trusted third parties.  For
example, the approach by \citet{spot} assumes the following trust
model: each party outsources their data to a cloud provider and then
runs the protocol.  Sometimes they even use the same cloud provider.
Clearly, this cloud provider is a real-world trusted third party,
since it has access to both parties' data and could completely subvert
the security of the protocol.  However, \citeauthor{spot}, in our
opinion rightly, assume that the PSI protocol adds a layer of
security, since the data is not revealed to the other party hosted by
the cloud provider.  Yet, given such a trusted cloud provider, it is
easy to extend the model, such that the cloud provider also creates
the output of the offline phase (as our OLE tuples).  This will
significantly reduce the cost of our offline phase and hence our
entire protocol significantly below the cost of \citeauthor{spot}'s
protocol without changing the trust assumptions.

\subsubsection{Trusted Execution Environments} 
Secure hardware-supported implementations of the offline phase are also possible.
For example, Microsoft's Azure cloud offers trusted execution environments (TEE) such as Intel's SGX.
Microsoft's Azure cloud's TEEs are also used to privately intersect the set of client's address book with the set of the Signal user base~\cite{signal}.
Using trusted hardware (and trusted software to run in the hardware), we will show that the cost of our offline phase is also very low.
For fair comparison, we emphasize that this is an additional security assumption not made by the protocols with which we compare our performance~\cite{cm,kkrt,spot,vole,crypto21}.
However, the main advantage of our protocol stems from the shorter online time which does not make additional security assumptions.
Hence, we still consider this a fair comparison and Signal's use of TEEs underpins the acceptability of this assumption.
The advantage of only running the offline phase in the TEE instead of the entire PSI protocol is that the trusted software (which is difficult to produce and verify) can be re-used across many secure computation protocols, such as SPDZ, and the memory requirement is small opposed to the large data set sizes for PSI that may need to be held in memory and randomly accessed in a hash table.
We need a TEE at both parties, Alice and Bob, for maximum efficiency which is different from a single trusted third party.

\subsubsection{Cryptographic Protocols} We will also describe two protocols for securely pre-computing OLE tuples using only cryptographic assumptions (lattice-based homomorphic encryption and OT). 
These are variations of previous work on product sharing.
If a party is not willing to accept additional security assumptions during the offline phase (only), they would these offline phases and compare them as part of the whole protocol with related work.
We emphasize that our performance gain during the online phase is unaffected by any performance loss during the combined online and offline phase.
First, observe the relation between OLE tuples and Beaver's multiplication triples.
In Beaver triples, equation $c=a\cdot b$ holds with $a=a_0+a_1$, $b=b_0+b_1$, $c=c_0+c_1$, and one party holds $(a_0,b_0,c_0)$, and the other holds $(a_1,b_1,c_1)$. 
As $(c_0+c_1)=(a_0+a_1)\cdot{}(b_0+b_1)=a_0b_0+a_0b_1+a_1b_0+a_1b_1$, the only information parties have to interactively compute is $a_0b_1$ and $a_1b_0$. 
That is, the two parties compute two product sharings~\cite{product-sharing}, where the factors stem from the parties, and the products are secret shared between the parties. 
The computation of product shares/OLE lies at the heart of offline phases to compute Beaver multiplication triples.
Highly practical MPC systems typically use one of two approaches to realize product sharing: lattice-based, somewhat-homomorphic encryption (e.g., SPDZ~\cite{spdz} and Overdrive~\cite{overdrive}) or Oblivious Transfer (e.g., MASCOT~\cite{mascot}). Both approaches were initially mentioned by \citet{gilboa}.
As the cryptographic protocols for product sharing and consequently Beaver triples in today's general MPC frameworks target (expensive) malicious security, they are significantly slower than lattice-based encryption and OT implementations that we will describe in the following.

All protocols we describe in detail below have computational
complexity linear in the number of tuples. The protocol using TEEs
features constant communication complexity, while the three others
have linear communication complexity.

{\noindent{\bf Discussion.}} Improvements in the underlying, adapted
cryptographic protocols and techniques we build on can further enhance
the performance of our offline phase.
Recent advances in pseudo-random correlation generators (PRCG)~\cite{boyle1} theoretically enable the generation of OLE tuples and even Beaver's multiplication triples~\cite{beaver} with very little communication effort.
VOLE-PSI~\cite{vole} uses a PRCG for vector-OLE (VOLE), and an implementation of vector-OLE is available~\cite{vector-ole}.
However, vector-OLE, as used in VOLE-PSI~\cite{vole}, does not apply to the OLE tuples we require. 
Adapting existing implementations of vector-OLE to our needs leads to inferior performance than our implementations described in Section~\ref{sec:precomputation}.
Hence, we delay the development of improved PRCG techniques for the offline phase to future work.

We now present details of our four protocols for the offline phase.

\subsection{Third Trusted Party (Dealer)}
\label{sec:thirdparty}

A naive implementation of the TTP would create the entire OLE tuples at the TTP and then distribute their respective pairs $(s_A,r_{A,1},\ldots,r_{A,\beta})$ to Alice and $((r_{B,1},s_{B,1}),\ldots,(r_{B,\beta},s_{B,\beta}))$ to Bob.  
However, we found the following implementation using cryptographic assumptions to be more efficient even in current high-speed networks.
The TTP chooses two random seeds $R_A$ and $R_B$.
TTP uses $R_A$ as a seed in a pseudo-random number generator (PRG) to generate $s_A$ and uses $R_B$ to generate $\beta$ (for each $s_A$) $s_{B,i}$s and $r_{B,i}$s.
TTP then computes $r_{A,i} = (s_A + s_{B,i}) / r_{B,i}$ and sends $R_A$ and all $r_{A,i}$ to Alice and $R_B$ to Bob.
Alice and Bob use the same PRG to generate $s_A$, $s_{B,i}$, and $r_{B,i}$, respectively. This process repeats for all OLE tuples required.

\subsection{Trusted Execution Environments}
\label{sec:teeoffline}
Some cloud providers, such as Microsoft's Azure, offer the use of Trusted Execution Environments (TEE) such as Intel's SGX.

A straightforward approach for PSI would be that Alice and Bob compute the entire PSI protocol in a TEE.
However, this has a number of disadvantages.
Programs to be executed in a TEE need to be hardened against side channel attacks and be free of any software vulnerability, such that the party controlling the host computer cannot infer the computation done in the TEE.
This has significant costs for both parties -- one developing the program and the other verifying the correct development (including compilation).
Furthermore, in current TEEs, e.g., Intel's SGX, the encrypted memory for the enclave is small (96~MBytes).
Such a small memory size requires memory paging even for the moderately-sized data sets we consider, and paging quickly becomes a performance bottleneck, particularly under random access in a hash table~\cite{scone}.
Finally, splitting the PSI program into a secure part that runs in the TEE and an insecure part that has efficient memory access comes with significant implementation challenges.
One would have to authenticate calls from the insecure part to the secure part which may require additional security assumptions.

\ignore{
Even without authentication of calls the performance of an insecure version that splits the PSI program can be disappointing.

Consider the following strawman construction.
It computes messages authentication codes (MAC) using a key securely agreed between the enclaves at each party inside those enclaves and performs the intersection on the MACs in the unsecure parts.
This construction is $29\%$ to $109\%$ slower in our cloud settings for data set sizes between $2^{20}$ and $2^{26}$ than our protocol with an offline phase in SGX even when the MACs are only sent from Bob to Alice.
We attribute this disappointing performance to the computational cost of the MACs which is much higher than our simple field operations and generation of pseudo-random numbers.
}

Instead, we suggest to only generate OLE tuples in the TEE. Using the generation of OLE tuples as the program that is to be securely developed, verified, and deployed to the TEE has several advantages.
First, a program for OLE tuple generation is very small and hence easier to harden.
Second, it can be potentially be re-used among several applications using pre-computation of correlated randomness, such that its development costs amortize.
Finally, OLE tuples can be sampled and communicated to untrusted program with constant memory requirements.


A naive implementation using TEEs would use one enclave and use that to generate all tuples and distribute them securely over an encrypted channel (e.g.,~TLS).
However, this incurs significant communication costs, since all tuples of one party need to be sent from the TEE's host computer to that party.
A more communication-efficient implementation uses two TEEs, one at Alice's site and one at Bob's.
Both, Alice and Bob load the same program into their TEE which contains a public key of each party for authentication.
Each party locally attests its deployment and verifies  the remote deployment at the other party's TEE, such that the integrity of the programs at start is assured.
The programs then jointly choose a random seed.
Using that seed, they secretly generate the same OLE tuples using a pseudo-random number generator (PRG).
Each party locally authenticates to its TEE and the TEE releases the pairs of the OLE tuples for that party.
This implementation has only small, constant communication cost for the establishment of the joint seed.

\subsection{Lattice-Based Homomorphic Encryption (LBE)}
\label{sec:somewhat}

We can precompute the OLE tuples using LBE.
Alice chooses $m$ keys over plaintext fields $\mathbb{F}_{q_i},0 < i \leq m$.
If our prime $Q$ (as used in the comparison protocol) is a possible plaintext size for the LBE scheme, we set $m=1$ and $q_1 = Q$.
If prime $Q$ is larger than the possible plaintext sizes for a chosen security parameter, then we set $m > 1$ and for a statistical security parameter $\lambda$ we choose
$$Q' = \prod_{i = 1}^{m} q_i > Q^2 2^{\lambda}$$

Alice chooses a uniform $s_A \in \mathbb{F}_Q$.
Let $E_i()$ denote the homomorphic encryption under the $i$-th key (with plaintext modulus $q_i$).
Alice encrypts as 
$$c_i = E_i(s_A \bmod q_i)$$ 
for $1 \leq i \leq m$ and sends all $c_i$ to Bob.
Bob chooses uniform $s_B \in \mathbb{F}_Q$ and $r_B \in \mathbb{F}^*_Q$.
If $m > 1$, Bob uniformly chooses $u$ in $\mathbb{F}_{2^{\lambda}}$, else Bob sets $u=0$.
Then he computes 
$$d_i = (c_i + E_i(s_B)) ((r_B^{-1} \bmod Q) \bmod q_i) + E_i(uQ \bmod q_i)$$
and sends all $d_i$ to Alice.
Alice decrypts $d_i$ and if $m > 1$, Alice uses the Chinese Remainder Theorem to recover $v = uQ + (s_A + s_B) / r_B \in \mathbb{F}_{Q'}$.
Alice sets $r_A = v \bmod Q$.

We can re-use the optimization of the comparison protocol (see Section~\ref{sec:protocol}) and Alice sends one $s_A$, but Bob chooses $r_B$, $s_B$ and computes $r_A$ for each pair of comparison with the same element $x$.
This reduces the number of ciphertexts that Alice needs to send.
Furthermore, ciphertext operations can be easily batched over vectors of plaintexts, a common SIMD technique in LBE.

Homomorphically encrypted ciphertexts conceal Alice's inputs from Bob, but not vice-versa.
Specifically, the noise level in the ciphertexts after the protocol might reveal information about Bob's inputs to Alice. To prevent this, there exists techniques to provide \emph{circuit privacy}. In short, circuit privacy allows the function that Bob evaluates ($f(x) = ax +b$ in this case) to remain hidden from Alice. \emph{Noise flooding} is one technique to make the noise level of the ciphertext after the protocol statistically indistinguishable from the noise level a freshly encrypted ciphertext. Before communicating the results back to Alice, Bob uses such techniques to make the noise level of the ciphertexts statistically indistinguishable from fresh ciphertexts.

\begin{proposition}
The above OLE tuple precomputation protocol using LBE is {\em correct}, i.e., $r_A r_B = s_A + s_B$.
\end{proposition}

This trivially follows from the construction of the protocol.

\begin{theorem}
The above OLE tuple precomputation protocol using LBE is {\em secure}, i.e., there exists a simulator of Bob's and Alice's view.
\end{theorem}

\begin{proof}
{\em Bob's simulator}:
Bob receives $m$ IND-CPA secure homomorphic ciphertexts.

{\em Alice's simulator}:
In case $m = 1$, Alice's messages can be simulated by the message $r_A$, i.e., part of the output of the protocol.
In case $m > 1$, Alice's messages can be simulated by a choosing random number $u' \in \mathbb{F}_{2^{\lambda}}$ and the message $r_A + u'Q$.
The statistical indistinguishability of these views has been proven by \citet{thorbek}.
Due to circuit privacy, the noise level of the ciphertexts are statistically indistinguishable from fresh ciphertexts.
\end{proof}

\subsection{Oblivious Transfer}
\label{sec:otoffline}
We use a second, alternative cryptographic approach to prepare OLE tuples that is based on \citeauthor{gilboa}'s product sharing~\cite{gilboa}. This is technique is implicitly used as a sub-routine in various general MPC frameworks, see, for example, \citet{mp-spdz} for an overview.

\paragraph{Intuition} To generate an OLE tuple, Alice chooses $r_A$ randomly from $\mathbb{F}_Q$, and Bob randomly chooses $r_B$ from $\mathbb{F}^*_{Q}$.
Let $\ell=\log{Q}$ and $r_B[i]$ be the $i$-th bit of $r_B$.
The main idea is to rewrite $r_A\cdot{}r_B$ using the multiplication formula $$r_A\cdot{}r_B=\sum^{\ell}_{i=1}r_B[i]\cdot{}r_A\cdot{}2^{i-1}.$$  
Alice and Bob run $\ell$ instances of 1-out-of-2 OT. 
In each instance, receiver Bob uses $r_B[i]$ as choice bit, and sender Alice inputs $0$ and $r_A\cdot{}2^{i-1}$.
By summing up his OT output, Bob gets $r_A\cdot{}r_B$.
To restrict Bob to only get a share of $r_A\cdot{}r_B$, Alice offsets each input to the OT by a random $\rho_i$.
Alice's share $s_A$ of $r_A\cdot{}r_B$ is then the sum of the $\rho_i$, and Bob's share $s_B$ is the sum of the OT output.

\paragraph{Details} More formally, to generate an OLE tuple $(r_A,r_B,s_A,s_B)$, Alice and Bob run the following protocol.

\begin{enumerate}
\item Alice randomly chooses $r_A$ from $\mathbb{F}_Q$ and $\ell$ random values $\rho_i$ from $\mathbb{F}_Q$. Bob randomly chooses $r_B$ from $\mathbb{F}^*_Q$.
\item Alice and Bob run $\ell$ instances of 1-out-of-2 OT. In round
  $i$, sender Alice inputs $-\rho_i$ and
  $r_A\cdot{}2^{i-1}-\rho_i$. Receiver Bob inputs choice bit $r_B[i]$ and
  receives output $o_i$.
\item Alice sets $s_A=\sum^\ell_{i=1}{\rho_i}$, and Bob sets
$s_B=\sum^\ell_{i=1}o_i$.\label{rand:choose}
\end{enumerate}

The correctness and security of this protocol has been shown by \citeauthor{gilboa}~\cite[Section 4]{gilboa}.

Recall that the optimized communication during the online phase
(Section~\ref{sec:commopti}) requires only one $s_{A,i}$ for $\beta$-many
$(s_{B,i,j},r_{A,i,j},r_{B,i,j})$. That is, we need to compute the tuple
$(s_{A,i},\{r_{A,i,j},r_{B,i,j},s_{B,i,j}\}_{j=1\ldots\beta})$ such
that $s_{A,i}+s_{B,i,j}=r_{A,i,j}r_{B,i,j}$ for $j\leq\beta$. To
enable OT-based offline preparation of communication optimized OLE
tuples, we consequently must institute one additional change. In the above protocol,
instead of choosing $\ell$ random $\rho$ values and setting
$s_{A,i}=\sum_{\iota=1}^\ell{\rho_\iota}$ in Step~\ref{rand:choose},
Alice randomly chooses one $s_{A,i}$ and $\beta\ell$-many random
$\rho_{j,\iota}$ such that
$s_{A,i}=\sum_{\iota=1}^{\ell}{\rho_{j,\iota}}$ for each $j\leq\beta$. While computing the
$j^{\text{th}}$ product sharing, during round $\iota$, Alice uses $-\rho_{j,\iota}$ and
 $r_{A,i,j}\cdot{}2^{\iota}-\rho_{j,\iota}$ as her input to the OT.
 
We omit proofs of correctness and security, since they are simple corollaries of the theorems by \citet{gilboa} and to abide by page length restrictions.

\section{Evaluation}
\label{sec:evaluation}


\begin{table*}[!tb]\centering
\caption{\label{tab:para}Parameter analysis for our system with statistical security parameter $\lambda = 40$. $n$: number of elements, ${\log{Q}}:$ number of bits to compare, $s:$ size of stash, $\frac{\text{Bit}}{\text{element}}:$ online communication cost per element (in Bit, see text)}
\label{tbl:hashpar}
\begin{tabular}{c|cccc|ccc|ccc}
&\multicolumn{4}{c|}{\multirow{2}{*}{$k=2$, $\alpha=2.4n$}}  &\multicolumn{3}{c|}{$k=3$, $\alpha=1.27n$ }&\multicolumn{3}{c}{$k=4$, $\alpha=1.09n$ } \\
&&&&&\multicolumn{3}{c|}{(no stash $s=0$)}&\multicolumn{3}{c}{(no stash $s=0$)}
  \\$n$&$\beta$ & ${\log{Q}}$ &$s$&$\frac{\text{Bit}}{\text{element}}$&$\beta$ & ${\log{Q}}$&$\frac{\text{Bit}}{\text{element}}$&$\beta$ & ${\log{Q}}$&$\frac{\text{Bit}}{\text{element}}$
     \\\hline
    $2^{20}$ &19&13&3&663& 28&14&\bf{516}&33&15&556
\\    $2^{22}$&20&11&3&588& 28&12&\bf{442}&34&13&496
    \\    $2^{24}$&20&9&2&472 &29&10&\bf{381}&35&11&432
    \\    $2^{26}$ &21&7&2&384&29&8&\bf{305}&35&9&353

\end{tabular}
\end{table*}







\begin{table*}[!tb]\setlength{\tabcolsep}{3pt}
\centering  
\caption{\label{tbl:bitperelement} Theoretical communication cost in Bit per element of our scheme vs. related work, $\lambda=40,\kappa=128$. GPRTY21~\cite{crypto21} and CGS22~\cite{pets22} specify only asymptotic communication complexity, so we extrapolate from their benchmarks.}
  \begin{tabular}{c|ccccc|cccccc|cc}
    \multirow{2}{*}{$n$}& \multicolumn{5}{c|}{Ours ($k=3, \alpha=1.27n$)}&{PSTY19} & {KKRT16} & {PRTY19}&{CM20}&\multicolumn{2}{c|}{RS21~\cite{vole}}& {GPRTY21}&{CGS22 ($\mathsf{PSM}_2$)}
    \\&\bf{online}& total SGX& total TTP& total LBE&total OT&\cite{opprf}&\cite{kkrt}&\cite{spot}&\cite{cm}&\bf{online}&total&\cite{crypto21}&\cite{pets22}
    \\\hline
    $2^{20}$ &516&516&1014&2922&78179&20079 &972&{514}&688&\bf{394}&\bf{419}&(774)&(8856)
\\    $2^{22}$&442&442&869&2956&65303&20079 &980&516&691&\bf{393}&\bf{400}&(943)&(8870)
\\    $2^{24}$&\bf{381}&\bf{381}&749&2893&54889&20079 &988&518&694&396&{398}&(1622)&(8926)
\\    $2^{26}$&\bf{305}&\bf{305}&600&2928&42733&20079 &997&520&697&400&{400}&(4338)&(9150)
  \end{tabular}
\end{table*}

\begin{table*}[!tb]\centering
  \caption{\label{tbl:offline}Offline benchmarks. CPU time: computation time only, Total Offline Time: end-to-end total time (including communication time, waiting for other party etc.) until parties have all OLE tuples. DNF: did not finish in 15min.}
  \begin{tabular}{c|c|S[table-auto-round,table-format=3.1]|S[table-auto-round,table-format=3.1]S[table-auto-round,table-format=3.1]S[table-auto-round,table-format=3.1]S[table-auto-round,table-format=3.1]S[table-auto-round,table-format=3.1]S[table-auto-round,table-format=3.1]|c}
    \multirow{2}{*}{$n$}&&{CPU Time}&\multicolumn{6}{c|}{Total Offline Time (s)} & Communication
    \\&&{(s)}&{5 GBit/s}&{1 GBit/s}&{100 MBit/s}&{10 MBit/s}&{Intra-Cont.}&{Inter-Cont.}&{(MB)}
    \\\hline{}\multirow{4}{*}{$2^{20}$}&SGX& 0.218 & 0.258& 0.258& 0.258&\fpeval{0.258}&\fpeval{1.58*0.258}&\fpeval{1.58*0.258}& $1.5\cdot{}10^{-5}$                                         
    \\&TTP&0.152798&0.253071&0.649439&5.23828&52.150 &\fpeval{1.58*2.01}&\fpeval{1.58*2.281}&   62.23
    \\&LBE&10.687&14.7091&18.0819&55.5732&430.263&\fpeval{1.58*29.404}&\fpeval{1.58*30.3487}&482
\\&OT&13.68&44.71&205.67&{DNF}& {DNF}&{DNF}&{DNF}&23896
                                               
    \\\hline{}\multirow{4}{*}{$2^{22}$}&SGX & 1.113  & 1.513  & 1.513  & 1.513&1.513&\fpeval{1.58*1.513}&\fpeval{1.58*1.513}&$1.5\cdot{}10^{-5}$                                         
    \\&TTP&0.631&0.984&2.377&18.759&182.432&\fpeval{1.58*7.253}&\fpeval{1.58*7.669}& 213
    \\&LBE&42.4072&58.5636&71.5634&221.550&{DNF}&\fpeval{1.58*117.038}&\fpeval{1.58*121.101}&1915
\\&OT&50.244&150.380&703.034&{DNF}&{DNF}&{DNF}&{DNF}&81930

    \\\hline{}\multirow{4}{*}{$2^{24}$}&SGX &7.819&7.859&7.859&7.859&7.859&\fpeval{1.58*7.859}&\fpeval{1.58*7.859}&$1.5\cdot{}10^{-5}$                                         
    \\&TTP&2.522&3.757&8.834&65.855&636.396&\fpeval{1.58*25.6853}&\fpeval{1.58*27.226}&   737.6
    \\&LBE&175.571&240.322&289.764&{DNF}&{DNF}&\fpeval{1.58*481.368}&\fpeval{1.58*498.341}&7951
\\&OT&178.862&519.882&{DNF}&{DNF}&{DNF}&{DNF}&{DNF}&  282855

    \\\hline{}\multirow{4}{*}{$2^{26}$}&SGX &39.080&39.120&39.120&39.120&39.120&\fpeval{1.58*39.120}&\fpeval{1.58*39.120}&$1.5\cdot{}10^{-5}$                                         
    \\&TTP&11.536&15.5051&31.4193&214.675&{DNF}&\fpeval{1.58*86.087}&\fpeval{1.58*91.085}&     2357
    \\&LBE&{DNF}&{DNF}&{DNF}&{DNF}&{DNF}&{DNF}&{DNF}&{DNF}
\\&OT&576.836 &{DNF}&{DNF}&{DNF}&{DNF}&{DNF}&{DNF}&905135

  \end{tabular}
\end{table*}

\begin{table*}[!tb]\centering\setlength{\tabcolsep}{2pt}
  \caption{\label{tbl:online} Online benchmarks. Time in seconds, communication in MByte. DNF: did not finish in 15~min or crashed. ``---'': source code not available. CPU: time for computation only, Online Total: end-to-end total time (including communication time, waiting for other party etc.) until intersection has been computed.  Values in ``$()$'' are taken from original papers due to lack of source code (\cite{vole,crypto21}). 
  GPRTY21~\cite{crypto21} benchmark with $4.6$~GBit/s, $260$~Mbit/s, $33$~MBit/s bandwidth.}
  \resizebox{\textwidth}{!}{\begin{tabular}{@{}cc|S[table-auto-round,table-format=3.1]S[table-auto-round,table-format=4.1]|S[table-auto-round,table-format=3.1]c|cc|S[table-auto-round,table-format=3.1]c|cccc|cS[table-auto-round,table-format=3.1,detect-weight,mode=text]|c|S[table-auto-round,table-format=3.1,detect-weight,mode=text]S[table-auto-round,table-format=3.1,detect-weight,mode=text]S[table-auto-round,table-format=3.1,detect-weight,mode=text]S[table-auto-round,table-format=3.1,detect-weight,mode=text]|c@{}}
    \multicolumn{10}{c|}{}&\multicolumn{4}{c|}{RS21~\cite{vole}}&\multicolumn{3}{c|}{}&\multicolumn{5}{c}{{{Ours Online+Offline}}}
    \\& &\multicolumn{2}{c|}{KKRT16~\cite{kkrt}} & \multicolumn{2}{c|}{PRTY19~\cite{spot}} &\multicolumn{2}{c|}{GPRTY21~\cite{crypto21}} &\multicolumn{2}{c|}{CM20~\cite{cm}} & \multicolumn{2}{c}{Online}& \multicolumn{2}{c|}{Online+Offline}&\multicolumn{3}{c|}{}&\multicolumn{4}{c|}{Time}&\multirow{3}{*}{{Comm}}
     \\ \multicolumn{2}{c|}{Sec.~Assumption} &\multicolumn{2}{c|}{LBE} & \multicolumn{2}{c|}{LBE} & \multicolumn{2}{c|}{LBE} & \multicolumn{2}{c|}{LBE} & \multicolumn{4}{c|}{LBE} &\multicolumn{3}{c|}{\bf{Ours Online}}&\multicolumn{4}{c|}{Security~Assumption}&
   \\ {$n$}&{Bandwidth}&{Time} & {Comm}&{Time} & {Comm} &{Time}& {Comm}&{Time} & {Comm} &{Time} & {Comm}&{Time} & {Comm}&{CPU} &{\bf{Time}} & {\bf{Comm}}&{SGX}&{TTP}&{{LBE}}&{OT}&
                                                                                                                                                                                                                                     
    \\\hline{}\multirow{6}{*}{$2^{20}$}
    & 5~GBit/s&1.3&{\multirow{6}{*}{114}}&16.838&{\multirow{6}{*}{76}}&(5.8)&{\multirow{6}{*}{(97)}}&5.056&{\multirow{6}{*}{89}}&(4.4)&{\multirow{6}{*}{\ubold (53)}}&(5.4)&{\multirow{6}{*}{(54)}}&{\multirow{4}{*}{0.3}}&\ubold 0.421&{\multirow{6}{*}{ 64}}&0.7&0.7&15.1&45.1&\multirow{6}{*}{\begin{tabular}{l}SGX: 64\\TTP: \fpeval{64+62}\\LBE: \fpeval{64+482}\\OT: \fpeval{64+23896}\end{tabular}}
    \\& 1~GBit/s&1.389&&17.614&&---&&5.051&& {---}&& {---}&&&\ubold 0.830&&1.1&1.5&18.9&206.5&
    \\ & 100~MBit/s&10.264&&17.670&&(10.6)&&8.073&&(8.5)&&(9.9)&&&\ubold 5.872&&6.1&11.1&64.4&{DNF}&
    \\ & 10~MBit/s&98.976&&66.712&&(38.3)&&75.858&&(48.7)&&(54.4)&&&\ubold 55.859&&56.1&108.0&486.1&{DNF}&
                                                                                                                                                                                                  
    \\\cline{15-16}\cline{18-21}&Intra-Cont.&4.129&&26.012&&{---}&&8.838&&{---}&&---&&{\multirow{2}{*}{0.5}}&\ubold 2.747&&\fpeval{2.747+0.4}&\fpeval{2.747+3.2}&\fpeval{2.747+46.5}&{DNF}
    \\&Inter-Cont.&5.780&&26.187&&{---}&&9.243&&---&&---&&&\ubold 3.371&&\fpeval{3.371+0.4}&\fpeval{3.371+3.6}&\fpeval{3.371+48.0}&{DNF}
    
    \\\hline \multirow{6}{*}{$2^{22}$}
                          & 5~GBit/s&5.145&{\multirow{6}{*}{471}}&66.074&{\multirow{6}{*}{314}}&\multicolumn{2}{c|}{\multirow{6}{*}{---}}&25.222&{\multirow{6}{*}{358}}&(23.9)&{\multirow{6}{*}{\ubold (209)}}&(25.6)&\multirow{6}{*}{(210)}&{\multirow{4}{*}{1.1}}&\ubold 1.537&{\multirow{6}{*}{221}}&3.1&2.5&60.1&151.9&{\multirow{6}{*}{\begin{tabular}{l}SGX: 221\\TTP: \fpeval{221+213}\\LBE: \fpeval{221+1915}\\OT: \fpeval{221+81930}\end{tabular}}}
                                                                                                                                                                                                                                                    
    \\& 1~GBit/s&5.686&&66.364&&&&25.587&&{---}&&{---}&&&\ubold 2.996&&4.5&5.4&74.6&706.0&
    \\ & 100~MBit/s&42.199&&73.266&&&&31.671&&(40.7)&&(43.0)&&&\ubold 20.156&&21.7&38.9&241.7&{DNF}
    \\ & 10~MBit/s&408.275&&275.163&&&&309.429&&(199.0)&&(204.7)&&&\ubold 191.789&&193.3&374.2&{DNF}&{DNF}
                                                                                                      
    \\\cline{15-16}\cline{18-21}&Intra-Cont&14.196&&103.218&&&&41.892&&{---}&&{---}&&{\multirow{2}{*}{2.1}}&\ubold 7.589&&\fpeval{7.589+2.4}&\fpeval{7.589+11.5}&\fpeval{7.589+184.9}&{DNF}
    \\&Inter-Cont.&19.849&&102.963&&&&{DNF}&&---&&{---}&&&\ubold 9.558&&\fpeval{9.558+2.4}&\fpeval{9.558+12.1}&\fpeval{9.558+191.3}&{DNF}

    \\\hline \multirow{6}{*}{$2^{24}$}
    & 5~GBit/s&21.402&{\multirow{6}{*}{1894}}&293.590&{\multirow{6}{*}{1271}}&\multicolumn{2}{c|}{\multirow{6}{*}{---}}&106.255&{\multirow{6}{*}{1442}}&(90.74)&{\multirow{6}{*}{(850)}}&(92.8)&\multirow{6}{*}{(851)}&{\multirow{4}{*}{5.0}}&\ubold 6.181&{\multirow{6}{*}{\ubold 762}}&14.0&9.9&246.5&526.1&{\multirow{6}{*}{\begin{tabular}{l}SGX: 762\\TTP: \fpeval{1500}\\LBE: \fpeval{762+7951}\\OT: \fpeval{762+282855}\end{tabular}}}
    \\& 1~GBit/s&22.562&&294.133&&&&108.210&&{---}&&&&&\ubold 11.438&&19.3&20.3&301.2&{DNF}&
    \\ & 100~MBit/s&169.483&&321.756&&&&128.034&&(156.4)&&&&&\ubold 70.783&&78.6&136.6&{DNF}&{DNF}&
    \\ & 10~MBit/s&{DNF}&&{DNF}&&&&{DNF}&&(814.2)&&&&&\ubold  662.669&&670.5&{DNF}&{DNF}&{DNF}&
                                                                                                
\\\cline{15-16}\cline{18-21}&Intra-Cont.&53.112&&460.501&&&&169.823&&{---}&&---&&{\multirow{2}{*}{9.6}}&\ubold  27.088&&\fpeval{27.088+12.4}&\fpeval{27.088+40.6}&\fpeval{27.088+760.6}&{DNF}
    \\&Inter-Cont.&126.734&&459.157&&&&{DNF}&&---&&---&&&\ubold 31.275&&\fpeval{31.275+12.4}&\fpeval{31.275+43.0}&\fpeval{31.275+787.4}&{DNF}
                                                                 
    \\\hline \multirow{6}{*}{$2^{26}$}
    & 5~GBit/s&86.292&{\multirow{6}{*}{7800}}&\multicolumn{2}{c|}{\multirow{6}{*}{DNF}}&\multicolumn{2}{c|}{\multirow{6}{*}{---}}&542.987&{\multirow{6}{*}{5769}}&\multicolumn{4}{c|}{\multirow{6}{*}{---}}&{\multirow{4}{*}{15.8}}&\ubold  19.247&{\multirow{6}{*}{\ubold  2438}}&58.4&34.8&{DNF}&{DNF}&{\multirow{6}{*}{\begin{tabular}{l}SGX: 2438\\TTP: \fpeval{2438+2357}\\LBE: {DNF}\\OT: \fpeval{2438+905135}\end{tabular}}}
    \\& 1~GBit/s&97.069&&&&&&549.920 &&&&&&&\ubold  35.766&&74.9&67.2&{DNF}&{DNF}&
    \\ & 100~MBit/s&697.316&&&&&&670.346 &&&&&&&\ubold  225.811&&264.9&440.5&{DNF}&{DNF}&
    \\ & 10~MBit/s&{DNF}&&&&&&{DNF} &&&&&&&{DNF}&&\multicolumn{4}{c|}{{DNF}}
                                                   
\\\cline{15-16}\cline{18-21}&Intra-Cont.&229.702&&&&&&{DNF}&&&&&&{\multirow{2}{*}{24.9}}&\ubold 96.194&&\fpeval{96.194+61.8}&\fpeval{96.194+136.0}&{DNF}&{DNF}
     \\&Inter-Cont.&553.786&&&&&&{DNF}&&&&&&&\ubold 99.448&&\fpeval{99.448+61.8}&\fpeval{99.448+143.9}&{DNF}&{DNF}
                                                
  \end{tabular}
}
  \end{table*}

We have implemented our protocol's online and offline phases in C++. The source code is available for download at~\hyperref[https://github.com/BlazingFastPSI/NDSS23]{\url{https://github.com/BlazingFastPSI/NDSS23}}.
Using this implementation, we have evaluated our protocol's communication cost (at the application layer in the network stack) and its running time.
We compare our implementation to the recent works of \citet{cm}, \citet{kkrt}, and \citet{spot} (``SpOT-low'') using their publicly available implementations.
We also compare to the recent VOLE-PSI protocol by \citet{vole} and OKVS by \citet{crypto21} using their measurements in their paper, since no implementation is publicly available. 
These two protocols are especially interesting, as they have low communication requirements (in bits per element).
Our protocol is a set intersection protocol, revealing the set intersection, but we also compare to state-of-the-art circuit PSI (or PSI analytics) protocols~\cite{opprf,pets22} to demonstrate the difference. 
VOLE-PSI requires an offline phase such as our protocol, but we concentrate on online timings and communication in this section. 
Note that VOLE-PSI by \citet{vole} can also be extended to PSI circuit PSI. 

We stress that there exists a large body of related work on practical
PSI, e.g., OT-based PSI~\cite{psiot,phasing}, but our experiments
only focus on the very recent works mentioned above as they typically outperform
other approaches.

All of the related work with which we directly compare uses purely cryptographic assumptions.
If a party is not willing to accept any additional security assumption during the offline phase only, the relevant comparison is with our offline phase implemented based on lattice-based encryption.
However, since we focus on the online time of our protocol which makes no additional security assumption, we also present the times for offline phases using additional security assumptions, such as TEEs or a trusted third party.

In all of our experiments, we use an element bit length of $\sigma = 32$.
This bit length is used in the evaluation of many related works~\cite{opprf,phasing,psiot,dualcuckoo,hashotex}. It is sufficient for several applications, including matching on phone numbers (``contact discovery'') or IPv4 addresses, without using hashing and hence without collisions.
Often, longer domains can be mapped to 32 bits without collisions using a public dictionary, since the largest data set sizes we consider are only $2^{26}$ (which already exceeds most previous work).

We consider the four implementations for our offline phase as described in Section~\ref{sec:precomputation}:
\begin{itemize}

\item We implement the TTP (Section~\ref{sec:thirdparty}), where the TTP derives all
  $(r_{A,i,j},r_{B,i,j},s_{A,i},s_{B,i,j})$, sends $R_A$ and all
  $r_{A_j}$ to Alice, and sends $R_B$ to Bob. Alice then computes the
  $(s_{A,i})$ and Bob the $(r_{B,i,j},s_{B,i,j})$.

\item We use Intel's SGX (Version 1) to implement the offline phase in a trusted execution environment (Section~\ref{sec:teeoffline}).

 \item  We use the Fan–Vercauteren cryptosystem~\cite{fv} for the offline phase using LBE (Section~\ref{sec:somewhat}). Our implementations employs Microsoft's SEAL library~\cite{sealcrypto}. There, we set the polynomial modulus degree and coefficient modulus bit length to be 8192 and 218, respectively. We also perform modulus switching to a 40-bit coefficient modulus. We do not implement the noise flooding technique for our experiments since the SEAL library does not support the necessary operations at the moment. However, based on the analysis of Castro et al.~\cite{fast-ole}, which is summarized in Lemma 4.2 of their paper, our parameter set provides circuit privacy for (v)OLE with more than 40 bits of statistical indistinguishability. The overhead of adding noise flooding to the protocol is negligible compared to the rest of the protocol (at most the cost of one homomorphic addition).

 \item We implement the OT-based offline phase (Section~\ref{sec:otoffline}), building on IKNP OT extensions~\cite{otextension} from EMP toolkit~\cite{emp-toolkit}.
   
\end{itemize}
\subsection{Theoretical communication cost}
\label{sec:communication}
We differentiate between online and offline communication cost.
First, we theoretically analyze the online communication cost and estimate the optimal parameters for the number of hash functions, hash table length, and the stash size in our scheme.
We set the statistical security parameter $\lambda = 40$ to minimize the probability of failure that cuckoo hashing exceeds  stash size $s$ or regular hashing exceeds the expected maximum number $\beta$ of  elements per bin.
Equation~(3) in Section 7.1 of \citet{phasing} provides an upper bound for $\beta$, i.e., the expected maximum number of elements per bin, given $n$ and $\alpha$.
Similarly, Section 3.2.2 in~\citet{psiot} provides estimations for the stash size based on extensive experimentation and extrapolation where necessary.
We compute the communication cost per element, i.e., the number of Bit per element, as 
\begin{equation}
\label{eq:onlinecomm}
\frac{\mathrm{Bit}}{\mathrm{element}} = \frac{\alpha\cdot(\beta + 1) + s \cdot (n+ 1)}{n}{\log{Q}}.
\end{equation}

Table~\ref{tbl:hashpar} summarizes our results.
We conclude that for set sizes between $2^{20}$ and $2^{26}$ using $k = 3$ has the lowest online communication cost and also does not require a stash.

Next, we theoretically analyze the offline and the total communication cost in bits per element and compare it to related work.
We consider the four possible implementations of the offline phase mentioned before: trusted execution environments, a trusted third party, lattice-based homomorphic encryption, and OT.
\begin{itemize}

\item The communication cost of using an SGX offline phase is $256$ bits to jointly choose a common seed. We did not include the time and communication for attestation, since it can be amortized over multiple runs.

\item The communication cost of the offline phase using a trusted
  third party is $128$ bits to Bob ($R_B$) and $128 + \alpha \beta \log{Q}$ bits
  to Alice ($R_A$ and $r_{A,i}$s).
  
\item Let $N$ and $q$ be the polynomial modulus degree and coefficient modulus in SEAL. Also, denote the coefficient modulus after modulus switching by $q_0$.
Each ciphertext sent from Alice to Bob has a size of $N\log_2 q$ bits and Alice can pack $N$ different $s_A$ values into one ciphertext.
In sum, Alice sends $N\log_2 q\cdot\lceil\frac{\alpha}{N}\rceil$ bits of data to Bob.
For each ciphertext received, Bob sends $\beta$ ciphertexts back to Alice.
Ciphertexts that are sent from Bob to Alice have a size of $2N\log_2 q_0$.
The total communication cost using lattice-based homomorphic encryption is $N \cdot (2\beta\log_2 q_0 + \log_2 q)\lceil\frac{\alpha}{N}\rceil$ bits.

\item For the OT-based offline phase, we need $\alpha\cdot\beta$
  product sharings. Each product sharing is implemented by $\log{Q}$
  instances of 1-out-of-2 OT extensions of length $\log{Q}$ bit
  strings. This results in a total of $m=\alpha\beta\log{Q}$ OT
  extensions with a total of $2m\log{}Q=\alpha\beta\cdot{}2\log^2{Q}$ Bit of
  communication.

  To bootstrap $m$ instances of IKNP OT extensions~\cite{otextension}, we
  need, first, to perform $\kappa$ (security parameter) base OTs of
  length $\kappa$~Bit, e.g., using the base OTs by \citeauthor{baseot}
  \cite[Figure~1]{baseot}.  When implemented over an elliptic
  curve, these base OTs require a total of
  $\kappa\cdot(2|P|+2\kappa)$~Bit of communication, where $|P|$
  denotes the bit length of a curve point. Then, we apply a
  transformation scheme, e.g., the one by \citet{alsz}, where the OT
  receiver sends $\kappa$ strings of length $m$~Bit to the
  receiver. So, this preparation phase requires in total
  $\kappa(2|P|+2\kappa)+\kappa{}m$~Bit of communication.
  
  In conclusion, the total communication cost for the OT offline phase
  (preparation plus extensions) computes to
  $2m\log{}Q+\kappa(2|P|+2\kappa)+\kappa{}m=\alpha\beta\log{Q}(2\log{Q}+\kappa)+2(\kappa^2+\kappa|P|)$~Bit.
  Note that a curve with $G$ points offers
  $\kappa\approx\log{\sqrt{\frac{G\cdot\pi}{4}}}$ Bit
  security~\cite{curves}, so for our security parameter $\kappa=128$
  in Table~\ref{tbl:bitperelement} we use a $256$~Bit curve and have
  $|P|=257$~Bit (using point compression).

\end{itemize}
Note that an alternative to standard IKNP OT extensions could be
recent Silent OT~\cite{silentot,silentottwo}.  Silent OT could reduce
communication costs at the expense of higher computational costs which
might be interesting for low bandwidth networks such as the 10~MBit/s
configuration we consider in Section~\ref{sec:runningtime}. We leave a
further analysis to determine the best OT extensions for different
network configurations to future work.

Using the above for the offline phases and
Equation~\ref{eq:onlinecomm} for the online phase, we can compute the
total communication cost in Bit per element of our PSI
protocol. Table~\ref{tbl:bitperelement} summarizes our results.  For
related work~\cite{kkrt,cm,spot,vole}, we use the equations from Table~1 (column 3) of
\citet{vole}.  As \citet{crypto21} and \citet{pets22} specify only
asymptotic communication complexity and also do not provide source code, we extrapolate their communication
costs from their benchmarks.


We conclude that our online communication cost and our total communication cost using Intel SGX is the lowest among related work for larger sets starting from $n=2^{24}$.

\subsection{Benchmarks}
\label{sec:runningtime}

\begin{table*}[tb]\caption{Online end-to-end time compared to circuit-PSI schemes. DNF: did not finish in 15~min or crashed. Value in ``()'' taken from original paper~\cite{opprf}}\label{tbl:circuitpsi}\centering
      \begin{tabular}{cc|S[table-auto-round,table-format=3.1,detect-weight,mode=text]c|S[table-auto-round,table-format=3.1,detect-weight,mode=text]c|S[table-auto-round,table-format=3.1,detect-weight,mode=text]c}
        &&\multicolumn{2}{c|}{PSTY19~\cite{opprf}}&\multicolumn{2}{c|}{CGS22 ($\mathsf{PSM}_2$)~\cite{pets22}}&\multicolumn{2}{c}{{\bf Ours}}
        \\$n$&Bandwidth&{{{Time}}}&{{Comm}}&{{{Time}}}&{{Comm}}&{\bf Time}&{\bf Comm}
        \\\hline\multirow{4}{*}{$2^{20}$}&5~GBit/s&37.533&\multirow{4}{*}{(2540)}&16.989&\multirow{4}{*}{1127}&\ubold 0.421&\multirow{4}{*}{\ubold 64}
        \\&1~GBit/s&52.524&&23.129&&\ubold 0.830
        \\&100~MBit/s&255.634&&107.175&&\ubold 5.872
        \\&10~MBit/s&{DNF}&&{DNF}&&\ubold 55.859

        \\\hline
        {$2^{22}$ -- $2^{26}$}&&\multicolumn{2}{c|}{{DNF}}&\multicolumn{2}{c|}{{DNF}}& {see Table~\ref{tbl:online}}
                                    
                                    \ignore{
        \\\hline\multirow{4}{*}{$2^{22}$}&5~GBit/s&\multicolumn{2}{c|}{\multirow{4}{*}{{DNF}}}&\multicolumn{2}{c|}{\multirow{4}{*}{{DNF}}}&\ubold 1.537&{\multirow{4}{*}{\ubold 221}}
        \\&1~GBit/s&&&&&\ubold 2.996
        \\&100~MBit/s&&&&&\ubold 20.156
        \\&10~MBit/s&&&&&\ubold 191.789
        
        \\\hline\multirow{4}{*}{$2^{24}$}&5~GBit/s&\multicolumn{2}{c|}{\multirow{4}{*}{{DNF}}}&\multicolumn{2}{c|}{\multirow{4}{*}{{DNF}}}&\ubold 6.181&{\multirow{4}{*}{\ubold 762}}
        \\&1~GBit/s&&&&&\ubold 11.438&
        \\&100~MBit/s&&&&&\ubold 70.783&
        \\&10~MBit/s&&&&&\ubold  662.669&
        
        \\\hline\multirow{4}{*}{$2^{26}$}&5~GBit/s&\multicolumn{2}{c|}{\multirow{4}{*}{{DNF}}}&\multicolumn{2}{c|}{\multirow{4}{*}{{DNF}}}&\ubold  19.247&{\multirow{4}{*}{\ubold  2438}}
        \\&1~GBit/s&&&&&\ubold  35.766&
        \\&100~MBit/s&&&&&\ubold  225.811&
        \\&10~MBit/s&&&&&{DNF}&
}
    \end{tabular}
  \end{table*}

  \ignore{
\begin{table*}[!tb]\centering\setlength{\tabcolsep}{3pt}
  \caption{\label{tbl:totaltime} Total runtime including offline phase in seconds, total communication (online+offline) in MByte. DNF: did not finish in 15~min or crashed. ``---'': source code not available. Values in ``$()$'' are taken from original papers (\cite{vole,opprf,crypto21}).  The PSTY19~\cite{opprf} implementation does not report communication. GPRTY21~\cite{crypto21} benchmark using $4.6$~GBit/s, $260$~Mbit/s, $33$~MBit/s bandwidth.}\centering
  \resizebox{\textwidth}{!}{\begin{tabular}{@{}cc|S[table-auto-round,table-format=3.1]S[table-auto-round,table-format=4.1]|S[table-auto-round,table-format=3.1]c|cc|S[table-auto-round,table-format=3.1]c|S[table-auto-round,table-format=3.1]c|S[table-auto-round,table-format=3.1]c|cc|S[table-auto-round,table-format=3.1,detect-weight,mode=text]S[table-auto-round,table-format=3.1,detect-weight,mode=text]S[table-auto-round,table-format=3.1,detect-weight,mode=text]S[table-auto-round,table-format=3.1,detect-weight,mode=text]|c@{}}

    \\\multicolumn{16}{c|}{}&\multicolumn{5}{c}{{{\bf Ours}}}  
\\&&\multicolumn{2}{c|}{KKRT16~\cite{kkrt}}&\multicolumn{2}{c|}{PRTY19~\cite{spot}}&\multicolumn{2}{c|}{GPRTY21~\cite{crypto21}}&\multicolumn{2}{c|}{PSTY19~\cite{opprf}}&\multicolumn{2}{c|}{CGS22~\cite{pets22}}&\multicolumn{2}{c|}{CM20~\cite{cm}}&\multicolumn{2}{c|}{RS21~\cite{vole}}&\multicolumn{4}{c|}{{\bf Time} with offline phase}&{{\bf Comm} with}
    \\$n$&{Bandwidth}&{{{Time}}}&{{Comm}}&{{{Time}}}&{{Comm}}&{Time}&{{Comm}}&{Time}&{Comm}&{Time}&{Comm}&{Time}&{Comm}&{Time}&{Comm}&{SGX}&{TTP}&{LBE}&{OT}&{offline phase}
                                                                                                                          
    \\\hline{}\multirow{4}{*}{$2^{20}$}
    & 5~GBit/s&1.3&{\multirow{4}{*}{114}}&16.838&{\multirow{4}{*}{76}}&(5.8)&{\multirow{4}{*}{(97)}}&37.533&{\multirow{4}{*}{(2540)}}&16.989&\multirow{4}{*}{1127}&5.056&{\multirow{4}{*}{89}}&(5.4)&{\multirow{4}{*}{ (54)}}&\fpeval{0.421+0.258}&\fpeval{0.421+0.253071}&\fpeval{0.421+14.7091}&\fpeval{0.421+44.71}&\multirow{4}{*}{\begin{tabular}{l}SGX: 64\\TTP: \fpeval{64+62}\\LBE: \fpeval{64+482}\\OT: \fpeval{64+23896}\end{tabular}}
    \\& 1~GBit/s&1.389&&17.614&&{---}&&52.524&&23.129&&5.051&& {---}&& \fpeval{0.830+0.258}&\fpeval{0.830+0.649439}&\fpeval{0.830+18.0819}&\fpeval{0.830+205.67}&
    \\ & 100~MBit/s&10.264&&17.670&&(10.6)&&255.634&&107.175&&8.073&&(9.9)&& \fpeval{5.872+0.258}&\fpeval{5.872+5.23828}&\fpeval{5.872+55.5732}&{DNF}& 
    \\ & 10~MBit/s&98.976&&66.712&&(38.3)&&{DNF}&&{DNF}&&75.858&&(54.4)&& \fpeval{55.859+0.258}&\fpeval{55.859+52.150}&\fpeval{55.859+430.263}&{DNF}&
    
    \\\hline \multirow{4}{*}{$2^{22}$}& 5~GBit/s&5.145&{\multirow{4}{*}{471}}&66.074&{\multirow{4}{*}{314}}&\multicolumn{2}{c|}{\multirow{4}{*}{---}}&\multicolumn{2}{c|}{\multirow{4}{*}{DNF}}&\multicolumn{2}{c|}{\multirow{4}{*}{{DNF}}}&25.222&{\multirow{4}{*}{358}}&(25.6)&{\multirow{4}{*}{ (210)}}& \fpeval{1.537+1.513}&\fpeval{1.537+0.984}&\fpeval{1.537+58.5636}&\fpeval{1.537+150.380}&{\multirow{4}{*}{\begin{tabular}{l}SGX: 221\\TTP: \fpeval{221+213}\\LBE: \fpeval{221+1915}\\OT: \fpeval{221+81930}\end{tabular}}}
    \\& 1~GBit/s&5.686&&66.364&&&&&&&&25.587&&{---}&& \fpeval{2.996+1.513}&\fpeval{2.996+2.377}&\fpeval{2.996+71.5634}&\fpeval{2.996+703.034}&
    \\ & 100~MBit/s&42.199&&73.266&&&&&&&&31.671&&(43.0)&& \fpeval{20.156+1.513}& \fpeval{20.156+18.759}& \fpeval{20.156+221.550}& {DNF}&
    \\ & 10~MBit/s&408.275&&275.163&&&&&&&&309.429&&(204.7)&& \fpeval{191.789+1.513}&\fpeval{191.789+182.432}&{DNF}&{DNF}&
                                                                 
    \\\hline \multirow{4}{*}{$2^{24}$}
    & 5~GBit/s&21.402&{\multirow{4}{*}{1894}}&293.590&{\multirow{4}{*}{1271}}&\multicolumn{2}{c|}{\multirow{4}{*}{---}}&\multicolumn{2}{c|}{\multirow{4}{*}{DNF}}&\multicolumn{2}{c|}{\multirow{4}{*}{{DNF}}}&106.255&{\multirow{4}{*}{1442}}&(92.8)&{\multirow{4}{*}{(851)}}& \fpeval{6.181+7.859}&\fpeval{6.181+3.757}&\fpeval{6.181+240.322}&\fpeval{6.181+519.882}&{\multirow{4}{*}{\begin{tabular}{l}SGX: 762\\TTP: \fpeval{1500}\\LBE: \fpeval{762+7951}\\OT: \fpeval{762+282855}\end{tabular}}}
    \\& 1~GBit/s&22.562&&294.133&&&&&&&&108.210&&{---}&& \fpeval{11.438+7.859}&\fpeval{11.438+8.834}&\fpeval{11.438+289.764}&{DNF}&
    \\ & 100~MBit/s&169.483&&321.756&&&&&&&&128.034&&(158.7)&& \fpeval{70.783+7.859}&\fpeval{70.783+65.855}&{DNF}&{DNF}&
    \\ & 10~MBit/s&{DNF}&&{DNF}&&&&&&&&{DNF}&&(819.9)&&  \fpeval{662.669+7.859}&{DNF}&{DNF}&{DNF}&
                                                                 
    \\\hline \multirow{4}{*}{$2^{26}$}
    & 5~GBit/s&86.292&{\multirow{4}{*}{7800}}&\multicolumn{2}{c|}{\multirow{4}{*}{DNF}}&\multicolumn{2}{c|}{\multirow{4}{*}{---}}&\multicolumn{2}{c|}{\multirow{4}{*}{DNF}}&\multicolumn{2}{c|}{\multirow{4}{*}{{DNF}}}&542.987&{\multirow{4}{*}{5769}}&\multicolumn{2}{c|}{\multirow{4}{*}{---}}&  \fpeval{19.247+39.12}&\fpeval{19.247+15.5051}&{DNF}&{DNF}&{\multirow{4}{*}{\begin{tabular}{l}SGX: 2438\\TTP: \fpeval{2438+2357}\\LBE: {DNF}\\OT: \fpeval{2438+905135}\end{tabular}}}
    \\& 1~GBit/s&97.069&&&&&&&&&&549.920 &&&&  \fpeval{35.766+39.120}&\fpeval{35.766+31.4193}&{DNF}&{DNF}
    \\ & 100~MBit/s&697.316&&&&&&&&&&670.346 &&&&  \fpeval{225.811+39.120}&\fpeval{225.811+214.675}&{DNF}&{DNF}
    \\ & 10~MBit/s&{DNF}&&&&&&&&&&{DNF} &&&&\multicolumn{4}{c|}{{DNF}}
                                                
                            \end{tabular}
                            }
  \end{table*}
}


We evaluate the runtime in two different environments and compare it to recent related work~\cite{kkrt,spot,opprf,cm,vole,pets22,crypto21}.
The two environments we consider are, first, a controlled environment where we can precisely adjust network bandwidth using Wondershaper~\cite{wondershaper}. This controlled environment is a workstation with a 3~GHz Intel Xeon W-1290 CPU and 64~GByte RAM.
The second environment comprises different data centers in the Microsoft Azure cloud.
There, we emulate an intracontinental scenario (``Intra-Cont.'') by benchmarking PSI protocols between two data centers on the US East coast and the US West coast.
We also emulate an intercontinental scenario  (``Inter-Cont.'') by benchmarking protocols between two data centers on the US East coast and Western Europe. Each data center is running our implementation in a standard Azure~H8 instance, i.e., on an Intel 3.2~GHz (3.6 GHz turbo frequency) Xeon E5-2667 CPU with 56~GByte RAM. 

Our implementation and all benchmarks use statistical security
parameter $\lambda=40$ and computational security parameter
$\kappa=128$.
 
\paragraph{Offline}
We first report on our evaluation of the offline phase in the controlled environment.
Note that although our protocol is highly parallelizable, all experiments both in the offline and online phases are run single threaded for fair comparison with related work. Our results are summarized in Table~\ref{tbl:offline}.

As expected, the TTP and SGX offline phases significantly outperform
the OT- and LBE-based ones. For very large values of $n=2^{26}$ or low
bandwidth networks ($10$~MBit/s), they do not complete within our set
time limit of 15~minutes, so we do not report measurement results. For
$n=2^{26}$, we can measure CPU time and communication requirements for
the OT-based offline phase, simply by releasing the network
throttle. This is not possible for the LBE-based offline phase, as
even computation does not complete in 15~min.
As we have theoretically analyzed in Section~\ref{sec:precomputation}, the dual TEE configuration in the SGX offline phase has a small constant communication cost for setting up the seed whereas all our other offline phases have linear communication cost.
Hence, the running time of the SGX offline phases is largely independent of the network environment different from our other offline phases.

We conclude that the offline phase based on SGX is by far the most
efficient in communication cost and total latency.  However, we note that it introduces an addition security assumption compared to the related work with which we directly compare.  Its communication
cost is much smaller than other implementations and the time to
communicate the tuples, even in plain, exceeds the computation cost
in a protected environment such as a trusted enclave.  The offline
phase implemented with a TTP, such as a cloud service
provider, could be scheduled for times of low utilization, thus
reducing its cost.  It turns out that, due to smaller communication
requirements, the offline phase using LBE is more efficient than the
OT-based ones. Moreover, it makes the least trust assumptions of all
offline phases.
Hence, LBE is the method of choice when relying only on cryptographic assumptions as related work \cite{cm,kkrt,spot,vole,crypto21} does.


\paragraph{Online} Next, we report on the running time and exchanged data of our online phase in comparison to related work~\cite{kkrt,spot,crypto21,cm,vole}.
We compare to the publicly available implementations where available and measure running times and the amount of data exchanged within the same environments.
All protocols are executed single-threaded as is common practice for fair comparison.
Only for VOLE-PSI~\cite{vole} and OKVS~\cite{crypto21}, we resort to their published numbers in the most comparable setting, since there is no publicly available implementation.
The faster of the two, VOLE-PSI~\cite{vole} was evaluated on an Amazon EC2 M5.2xlarge VMs, Intel Xeon Platinum 8175M 2.5 GHz (3.5 GHz turbo frequency), 32 GiB of RAM.
We chose Micrsoft Azure, since it provides access to TEEs, and consider our configuration the most similar to the Amazon one among the ones that were available.
For the SpOT-Light PSI protocol by \citet{spot}, we use the computation-optimized version, since the communication-optimized version had very high running times and performed worse in our tested environments.

The benchmarks are summarized in Table~\ref{tbl:online}. We measure
our implementation and implementations of related work and present
total, end-to-end running time and communication exchanged.  For our
scheme, we additionally show CPU-only time (including hashing, all comparisons, and
Alice's computation of the intersection) to allow a better
understanding of our scheme's performance.

On larger data sets, starting at $n=2^{24} = 16$ million elements, our online phase is  $2.4$ to $3.5$ times faster than the protocol by \citet{kkrt} and $1.2$ to $14.6$ times faster than the one by \citet{vole}, the two currently most competitive related works. 
At the same time, our protocol requires less communication than either one of them.
Kolesnikov et al.'s protocol~\cite{kkrt} performs best in high-performance network environments, which may be the future given that CPU performance currently increases slower than network performance, and it is already outperforming VOLE-PSI~\cite{vole} in the same-continent cloud setting (compared to the fastest time measured for VOLE-PSI).
Schoppmann and Rindal's protocol VOLE-PSI~\cite{vole} has the lowest communication cost for smaller sets with $n=2^{20}$ and $n=2^{22}$ and performs best in low-performance network environments.
VOLE-PSI has lower communication complexity ($O(n)$) compared to our protocol ($O(n \log n / \log \log n)$) but higher computation complexity ($O(n) \polylog n$) compared to our protocol's complexity of $O(n \log n / \log \log n)$.
However, in the common network settings as experienced in our real-world cloud experiments, we expect that the advantage in total latency of the online phase of our protocol compared to VOLE-PSI increases as data set sizes continue to grow, since the protocols in these environments are computation-bound and not communication-bound.





Although circuit-PSI provides a richer functionality than our PSI protocol, we compare communication cost and running times with the two most recent, cirucit-PSI protocols~\cite{opprf,pets22} in Table~\ref{tbl:circuitpsi}.
Not surprisingly, our protocol is much faster, and implementations of circuit-PSI protocols do not even scale beyond set sizes of $2^{20}$.

\paragraph{Total time}
When using lattice-based cryptography for our offline phase, Kolesnikov et al.'s protocol~\cite{kkrt} is the fastest in high-performance network environments and VOLE-PSI~\cite{vole} is the fastest in low-performance network environments.
Only when allowing additional security assumptions during the offline, such as TEEs, our protocol becomes the fastest in total running time of the combined offline and online phase.
Note that this does not affect our performance advantage during the online phase which is the focus of this work.

\section{Related Work}
\label{sec:related}

\citet{psi04} coined the term private set intersection (PSI) and since then progress has been fast and steady.
There exist many variants of PSI protocols.
We can distinguish between PSI protocols which reveal the intersection itself and PSI analytics (or circuit PSI) protocols which allow to compute functions over the intersection, such as the (differentially private) set intersection cardinality.
We can also distinguish between protocols secure in the semi-honest model or the malicious model.
However, we argue in Section~\ref{sec:semihonest} that this distinction is of minor practical importance in PSI (not necessarily PSI analytics) protocols.

\subsection{Industrial Work and Applications}

Google has deployed PSI for computing the value of ad conversion with Mastercard~\cite{psigoogle20,psigoogle15}.
Their protocol uses commutative elliptic curve encryption~\cite{meadows,commutative}, one of the first techniques to compute PSI.
The communication cost of this protocol is asymptotically optimal, very low in practice, and has only been recently matched by other protocols~\cite{opprf,vole}.
The computation cost of this protocol is rather high, but Yung mentions its flexibility as one of its advantages when deployed~\cite{psigoogle15}.
Their protocol is a PSI protocol revealing the intersection and they use extensions to accommodate certain functions, such as the differentially private set intersection cardinality~\cite{psigoogle20}.
These extensions are also applicable to our and other PSI protocols.
Their protocol is secure in the semi-honest, but not the malicious model.
These properties match our protocol, and our protocol is also conceptually simple as can been seen in Algorithm~\ref{alg:protocol}.

There exist other industrial implementations of PSI protocols such as by Facebook~\cite{facebook}, Microsoft~\cite{fhePSI1,fhePSI2}, VISA~\cite{vole} and VMWare~\cite{crypto21}.
The VOLE-PSI protocol supported by Google and VISA~\cite{vole} also has an offline phase like our protocol.
\citet{googlemal} have developed a maliciously secure, circuit PSI protocol for restricted functions, such as sum and cardinality.

A practical proposal is also to use three servers similar to our trusted third party setup \cite{kamara}.
In this setup, the parties share a symmetric key and send keyed hashes of the elements to a third party that compares them but does not have access to the key.
Compared to our construction -- even when using a trusted third party -- this setup has two disadvantages.
First, the third party is involved during the online phase whereas in our construction the third party is online involved in the offline phase.
Second, the online phase still requires cryptographic operations albeit only fast symmetric key ones whereas our online phase does not require any cryptographic operations and hence is significantly faster.

\subsection{PSI protocols}

PSI protocols compute comparison between the elements of the two sets.
The first approach by Meadows \cite{meadows} and Shamir \cite{commutative} is to compute the comparisons over pseudo-random functions (PRF) of the inputs.
Their construction still requires to compute one public-key operation per element.
The number of public-key operations can be reduced by using PCRGs or PRFs based on OT extensions \cite{otextension}.
Using OT extensions, it is possible to compute a small number of public-key operations and only use symmetric key operations for each element.
This significantly speeds up the entire protocol.
Several such constructions of oblivious PRFs (OPRF) exist \cite{otoprf,kkrt,malicious,spot,paxos,vole}.
These constructions can often be made secure in the malicious model with little overhead \cite{malicious,paxos,vole}.
The current state-of-the-art protocols for PSI with the best communication and computation cost \cite{kkrt,vole} are based on this construction.

We use a different construction than these protocols.
Instead of computing a PRF per element which has communication complexity $O(n)$, we perform a secure computation per comparison (which has higher communication complexity).
However, we are able to provide an incredibly fast implementation of comparisons (assuming an offline phase) using only four field operations and no cryptographic operations.
This improves the practical latency over medium-fast to fast networks as our experiments in Section~\ref{sec:runningtime} show.
Our protocol for comparisons is inspired by oblivious linear function evaluation (OLE) and Beaver multiplication triples \cite{beaver}.
\citet{psiole} already present a PSI protocol based on OLE, but use a complicated construction based on polynomials that cannot achieve our efficiency even when adapted to the semi-honest model.

When computing comparisons using a dedicated protocol instead of a plain comparison algorithm over PRFs, a challenge is to reduce the number of comparisons necessary, since each requires interaction.
Naively, there are $O(n^2)$ comparisons.
In a series of works, Pinkas et al.~\cite{hashotex,phasing,psiot,opprf} refined a strategy for Alice and Bob to hash their elements to common bins, such that only $O(n \log n / \log \log n)$ comparisons are necessary between pairs of common bins.
We follow this strategy as described in Section~\ref{sec:protocol}, but replace the comparison protocol with our own which makes this construction more efficient than OT extension-based OPRF ones.
Note that Pinkas et al.~also provide a different strategy for dual cuckoo hashing which leads to a circuit PSI protocols \cite{dualcuckoo}.

PSI analytics or circuit PSI protocols support a larger set of functionalities and hence applications, but are usually slower than PSI protocols.
A first PSI protocol based on the generic secure two-party computation protocol by Yao (garbled circuits) was presented by Huang et al.~\cite{yaoPSI}.
The most efficient current constructions base on OPRF protocols, but use a technique called oblivious, programmable PRF (OPPRF) \cite{origoopprf}.
The idea of an OPPRF is that the value of the PRF at given inputs can be fixed by the key holder. 
This allows the comparison result to be secret shared and available for subsequent secure computations over the intersection.
\citet{crypto21} generalize OPPRF to oblivious key-value stores (OKVS).
The most recent and most efficient circuit PSI protocol (under common wired network conditions) is VOLE-PSI \cite{vole} which follows this construction using OPPRFs.
VOLE-PSI also can be used as a PSI protocol (as any circuit PSI protocol) to which we compare favourably in Section~\ref{sec:evaluation}.

\subsection{Specialized PSI protocols}

There exist several variations of PSI and circuit PSI that cater for specific applications scenarios.
Circuit PSI performs a secure computation on the secret shared outputs of the PSI protocol.
However, the PSI protocols may be a part of a larger secure computation that does not start with a PSI protocol and consequently the inputs must also be secret shared.
\citet{secsharedPSI} present a variation of PSI for this setting.

It has been proposed to perform contact discovery in mobile phone messaging apps, such as Signal, using PSI.
In this case, one set, the cell phone users' one, is much smaller than the other, the messaging provider's one, i.e., the set sizes are imbalanced.
Modified OPRF protocols have been proposed for this setting \cite{imbalancePSI, imbalance}.
However, they require to store the PRFs of the larger set in the downloaded software with a key common to all users and consequently hinge on the security of this key.
Recent protocols without this restriction, which turn out to be also faster, use fully homomorphic encryption (FHE) \cite{fhePSI1,fhePSI2,fasterfhepsi}.
FHE-based constructions can also be used for circuit PSI \cite{rbc1}.

PSI with a semi-trusted third party has received attention under several different names.
It is beneficial when the communication between the two parties is restricted.
Kerschbaum introduced the concept in 2012 as outsourced PSI \cite{outsourcedPSI}.
\citet{arbiterPSI} named the third party an arbiter in 2013 and \citet{delegated} named it delegated PSI in 2015.
Some functions on the intersection, such as intersection cardinality, can be computed in this setting \cite{delegatedcard}.
Our offline phase using a trusted third party could be considered this setting, but our communication complexity between Alice and Bob remains large ($o(n)$).

PSI can also be computed among multiple, more than two, parties.
The protocol with the currently best complexity is by \citet{mulPSI} and the protocol with the best practical performance due to the extensive use of symmetric encryption is by \citet{origoopprf}.
Protocols with variants for the definition of intersection, such as the intersection of a subset with a size above a threshold, exist \cite{threshold,rbc2}.

PSI requires all elements to be unique within one set, i.e., there are no duplicates in each set.
This is different from database joins where elements in a relational table can be duplicated.
Database joins require different matching procedures~\cite{secsharedPSI,joins}.
PSI also performs only exact matches of the elements.
Private intersection with approximate matches is often referred to as private record linkage (PRL).
The challenge in PRL is to find the $O(n)$ matches among the $O(n^2)$ pairs where the hashing strategy used by PSI protocols does not apply.
Instead one can use locality-sensitive hashing \cite{dpprl} and then differentially privately pad the bins or locality-preserving hashing and use sliding windows \cite{prl}.

\section{Conclusions}
\label{sec:conclusions}

PSI protocols have been extensively studied in the literature, and progress in practical latency is only possible by improving constants.
In this paper, we present the first construction of a PSI protocol that is practically efficient and uses no cryptographic operations during the online phase.
Instead, all cryptographic operations are moved to an offline phase.
For the online phase, we present a new comparison protocol which consumes only one OLE tuple per comparison.
The online phase of our protocol outperforms the currently best related work~\cite{kkrt,spot,cm,vole,crypto21} by factors between $1.2$ and $3.5$ depending on the network performance.
There exist many different implementations for the offline phase which can be shared with other privacy-preserving protocols using precomputed, correlated, random secret shares.
Among others, we present a very efficient construction using TEEs, such as Intel's SGX.
The performance of  our PSI protocol with a TEE-based offline phase is currently unmatched.

\section*{Acknowledgements}
We gratefully acknowledge the support of NSERC for grants RGPIN-05849, IRC-537591, and the Royal Bank of Canada for funding this research.


%

\bibliographystyle{abbrvnat}
\bibliography{references}


\end{document}